\documentclass[final,12pt]{colt2026} %

\usepackage[utf8]{inputenc}
\usepackage[T1]{fontenc}
\usepackage{graphicx}
\usepackage{color}
\usepackage{enumerate}
\usepackage{verbatim}
\usepackage{thmtools,thm-restate}
\usepackage{wrapfig}
\usepackage{xspace}
\usepackage{amsfonts}
\usepackage{bbm}
\usepackage{dsfont}
\usepackage{algorithm, algpseudocode, caption, subcaption}

\definecolor{darkgreen}{rgb}{0,0.5,0}
\usepackage{hyperref}
\hypersetup{
    unicode=false,          %
    colorlinks=true,        %
    linkcolor=blue,          %
    citecolor=darkgreen,        %
    filecolor=magenta,      %
    urlcolor=cyan           %
}

\usepackage[disableredefinitions,roman]{complexity}
\usepackage[framemethod=tikz]{mdframed}
\global\mdfdefinestyle{myframe}{leftmargin=.75in,rightmargin=.75in,linecolor=black,linewidth=1.5pt,innertopmargin=10pt,innerbottommargin=10pt} 
\usepackage{mdframed}
\usepackage{framed}
\usepackage{nicefrac}

\date{}

\newtheorem{claim}[theorem]{Claim}
\newtheorem{fact}[theorem]{Fact}

\crefname{theorem}{Theorem}{Theorems}
\Crefname{lemma}{Lemma}{Lemmas}
\Crefname{alg}{Algorithm}{Algorithms}
\Crefname{claim}{Claim}{Claims}
\Crefname{fact}{Fact}{Facts}
\Crefname{infclaim}{Claim}{Claims}
\Crefname{observation}{Observation}{Observations}
\Crefname{invariant}{Invariant}{Invariants}
\Crefname{algorithm}{Algorithm}{Algorithms}

\newcommand{\eps}{\varepsilon}

\newcommand{\Ex}{\mathbb{E}}

\newcommand\numberthis{\addtocounter{equation}{1}\tag{\theequation}}

\renewcommand{\Pr}{\operatorname{{Pr}}}

\DeclareMathOperator*{\argmin}{arg\,min}

\newcommand{\da}{\operatorname{d}_{\mathcal{A}}}
\newcommand{\dak}{\operatorname{d}_{\mathcal{A}_k}}
\newcommand{\tv}{\operatorname{TV}}
\newcommand{\dhsq}{\operatorname{H}^2}
\newcommand{\hdist}{\operatorname{d}_{\mathcal{H}}}
\newcommand{\lin}{\operatorname{Ratio-Line}}
\newcommand{\rtag}[1]{\text{\quad $\left(\text{#1}\right)$}}
\newcommand{\opt}{\operatorname{OPT}}
\renewcommand{\R}{\mathbb{R}}

\allowdisplaybreaks[1]

\title[Density estimation for Hellinger via minimum-distance estimators]{Density estimation for Hellinger via minimum-distance estimators: mixtures of Gaussians, log-concave, and more}
\usepackage{times}

\coltauthor{%
 \Name{Spencer Compton} \Email{comptons@stanford.edu}\\
 \addr Stanford University
 \AND
 \Name{Jerry Li} \Email{jerryzli@cs.washington.edu}\\
 \addr University of Washington%
}

\jmlrproceedings{}{}
\jmlrvolume{}
\jmlrpages{}
\jmlryear{}
\jmlrworkshop{}
\editors{}
\makeatletter
\renewcommand*{\@titlefoot}{}
\makeatother

\begin{document}

\maketitle

\begin{abstract}
We study the task of density estimation, where we hope to accurately estimate a probability density from $n$ samples. A textbook method for density estimation in total variation distance is the \textit{minimum-distance estimator approach}, where we conclude both the algorithm and the analysis merely from bounding the VC dimension of a particular concept class (the so-called Yatracos class). 

While this technique has originally yielded sharp guarantees primarily for total variation distance, in this work we extend the minimum-distance estimator approach for learning within Hellinger distance. Our main observation is that we may produce an analogous recipe for Hellinger (where we only require bounding the VC dimension of a related concept class) by drawing connections to recent results yielding \textit{reverse data processing inequalities}. 

This recipe is flexible enough to accommodate fast algorithms originally designed for total variation distance; by modifying the approach of \cite{acharya2017sample} we conclude the first near-linear time algorithm for learning classes including univariate mixtures of log-concave densities and mixtures of Gaussians (with arbitrary variances), with near-optimal sample complexity.
\end{abstract}

\section{Introduction}
Density estimation is a cornerstone problem in statistical estimation.
Here, we receive $n$ samples from a distribution $f$, and our goal is to output an estimate $\hat{f}_n$ that is ``close'' to the true density of $f$. 
The goal is typically to design an estimator that is both statistically and computationally efficient: that is, the estimator should achieve minimax-optimal error, and the estimator should be efficiently computable; ideally, in linear (or near-linear) time in the number of samples.
To formalize the setting, there are two main desiderata one must specify:

\paragraph{What structural assumptions does $f$ satisfy?} Density estimation is typically most interesting when $f$ comes from some class of structured distributions, or when $f$ satisfies some shape constraints. 
There are many structural assumptions studied in the literature, including both parametric and non-parametric assumptions; see e.g. \cite{samworth2025nonparametric} for a recent survey.

A canonical example that will be particularly relevant to us is the class of mixtures of Gaussian distributions in one dimension, although our results also apply for many other classes of distributions as well.
A mixture of $k$ univariate Gaussians is a distribution over $\R$ with pdf given by
\[
p(x) = \sum_{i = 1}^k w_i N(\mu_i, \sigma_i^2) (x) \; ,
\]
where the $w_i$ are non-negative mixing weights satisfying $\sum_{i = 1}^k w_i = 1$, $\mu_i \in \R$ are the means of the components, and the $\sigma_i^2$ are the variances of the components.
Mixtures of Gaussians have been studied extensively in machine learning, statistics, and theoretical computer science, and are also still popular empirically as a way to model heterogeneous data in practice.

\paragraph{What error metric will we consider?}
To formalize the problem, one must specify a measure of closeness between distributions.
Different notions of distances are studied depending on the setting of interest (e.g., total variation distance, Hellinger distance, Kullback-Leibler divergence, $\chi^2$-divergence, Wasserstein distance, etc). 
Perhaps the most well-studied measure of closeness is the total variation distance: for two distributions $p, q$ defined over a shared probability space $\Omega$, the total variation distance between $p, q$, denoted $\tv$, is defined by
\begin{equation}
    \tv (p, q) = \frac{1}{2} \int_{\Omega} |p(x) - q(x)| \; .
\end{equation}
For our canonical example of mixtures of $k$ Gaussians in one dimension, essentially optimal estimators are known for this problem: work of~\cite{acharya2017sample} gives an estimator $\hat{f}_n$ which achieves error 
$
\tv (\hat{f}_n, f) = O \left( \sqrt{k / n} \right) ,
$
with high probability, and this rate is optimal.
Moreover, the estimator can be computed in nearly-linear time.
This estimator is based on a type of estimator known as a \emph{minimum-distance estimator}---a classic approach for obtaining sharp guarantees for closeness in total variation distance.

A natural question is whether or not one can achieve similar guarantees for other important statistical measures of closeness.
In this work, we will be focusing on (squared) Hellinger distance: 
\begin{definition}[Squared Hellinger distance]
    The squared Hellinger distance between $p$ and $q$ over a probability space $\Omega$ is defined as
    \begin{equation*}
        \dhsq(p,q) = \frac{1}{2} \int_\Omega \left(\sqrt{p(x)} - \sqrt{q(x)} \right)^2
    \end{equation*}
\end{definition}

While at first glance the definition may not look as intuitive as total variation distance, Hellinger distance is a workhorse underlying statistical learning theory. 
For example, the number of samples required to well-distinguish between any $p$ and $q$ is characterized by $\Theta(1/\dhsq(p,q))$ (and this cannot be expressed in terms of their total variation distance). 

Indeed, for estimation of many structured distribution classes (such as mixtures of log-concave distributions), Hellinger distance is arguably a more natural measure of closeness than total variation distance.
For one, for any probability distributions $p, q$, we have that
\begin{equation}
\label{eq:tv-to-hel}
\tv (p, q)^2 \leq 2 \dhsq (p, q) \; ,
\end{equation}
so in particular, up to constant factors, achieving squared Hellinger distance $\eps^2$ is strictly harder than learning to TV distance $\eps$, for any $\eps > 0$. 
For the types of distributions we consider in this paper, the statistical rates for learning in squared Hellinger distance that we will be able to achieve will be (up to logarithmic factors) quadratically better than the rates for learning in total variation, so our results subsume previous results for learning in total variation, up to logarithmic factors, for these classes of distributions.

To illustrate the difference, it is once again useful to return to the canonical example of mixtures of $k$ Gaussians.
Recall here the optimal rate for learning in total variation distance is given (up to constant factors) by $\sqrt{k / n}$.
But this means that, even when $k$ is constant, an estimator for learning a mixture of Gaussians is able to completely ignore any component with mixing weight less than $\sqrt{1 / n}$.
However, such components may be easily identifiable given $n$ samples from the mixture.
So in this sense, an estimator that achieves optimal total variation closeness for this problem may still be losing non-trivial information about the mixture.
In contrast, the optimal rate for learning a mixture of $k$ Gaussians in squared Hellinger distance is $O(k / n)$, which recovers the total variation rate by~\Cref{eq:tv-to-hel}, but to achieve such a rate, it is easy to show that one cannot ignore any component unless the mixing weight is less than $O(k / n)$, at which point the component is not even clearly identifiable from data. 

One can extend this argument to many other settings beyond mixtures of Gaussians: closeness in squared Hellinger also recovers the right rate for other useful tasks such as support estimation for uniform distributions, and tail estimation for log-concave distributions. Motivated by this discussion, in this work, we seek to understand how to design estimators for natural classes of structured univariate distributions which have optimal statistical and computational properties.

\subsection{Our Results}

Our main result is a new meta-algorithm for learning many natural classes of structured distributions in Hellinger distance, which both achieves nearly optimal statistical rates and runs in nearly-linear time, for the classes of distributions we consider, which include mixtures of Gaussians and mixtures of log-concave distributions, among many others.

For instance, when specified to the class of mixtures of $k$ Gaussians, our method yields the following guarantees:
\begin{theorem}\label{corr:fast-gauss}
Let $\mathcal{G}_k$ denote the class of mixtures of $k$ Gaussians. Given $n$ samples from a distribution $f$, there exists an algorithm that runs in $O(n \log^5(n))$ time and outputs an estimate $\hat{f}_n$ where with probability $1-\delta$,
\begin{equation*}
    \dhsq(f,\hat{f}_n) \lesssim \inf_{f_\theta \in \mathcal{G}_k} \log(n) \cdot \dhsq(f,f_\theta) \log(2k/\dhsq(f,f_\theta))  + \frac{k \log^2(n) \cdot (\log(n) + \log(2/\delta))}{n}.
\end{equation*}
\end{theorem}
Note in particular, up to logarithmic factors, our result achieves the optimal statistical rate, and the runtime of our algorithm is near-linear in the number of samples.
Moreover, our guarantee is also robust to some degree of model misspecification. We remark that these estimators are improper (they output a piecewise polynomial density, not a mixture of Gaussians).

Similarly, when specified to the class of mixtures of log-concave distributions (for a formal definition, see~\Cref{sec:preliminaries}), we obtain:
\begin{theorem}\label{thm:fast-log-concave}
Let $\mathcal{LC}_k$ denote the class of mixtures of $k$ log-concave distributions.
Given $n$ samples from a distribution $f$, there exists an algorithm that runs in $O(n \log^3(n))$ time and outputs an estimate $\hat{f}_n$ where with probability $1-\delta$,
\begin{equation*}
    \dhsq(f,\hat{f}_n) \lesssim \inf_{f_\theta \in \mathcal{LC}_k} \dhsq(f,f_\theta)  + \left( \frac{k \cdot (\log(n) + \log(2/\delta))}{n}\right)^{4/5}.
\end{equation*}
\end{theorem}
Once again, up to polylogarithmic factors, our algorithm achieves the optimal statistical rate, runs in nearly-linear time, and is robust.
We believe that by applying our formula, we can achieve similar nearly-optimal results for most, if not all, distribution families considered in prior work such as~\cite{acharya2017sample}.

\subsection{Our Techniques}
In many ways, our result can be thought of as the natural generalization of the minimum-distance estimator of \cite{yatracos1985rates}, and the framework developed in the line of work~\citep{chan2013learning,chan2014efficient,acharya2015fast,acharya2017sample} which culminated in a method by~\cite{acharya2017sample} for learning many classes of structured univariate distributions in total variation distance with optimal statistical rate and in nearly-linear runtime.
We briefly recap the ideas underlying the estimator in~\cite{acharya2017sample} here.

The key statistical idea in~\cite{acharya2017sample} is to use a distance they call the $\mathcal{A}_k$ distance as a proxy for closeness in total variation distance.
Here, for any parameter $k \geq 1$, we let $\mathcal{A}_k$ denote the set of subsets of $\R$ which are unions of at most $k$ intervals, and we define the $\mathcal{A}_k$ distance between two distributions to be
\begin{equation}\label{eq:ak}
    \dak (p,q) \triangleq \sup_{A \in \mathcal{A}_k} \left| \Pr_{p}[x \in A] - \Pr_{q}[x \in A] \right| \; .
\end{equation}
They then observe that: (1) if the underlying distribution can be well-approximated by a sufficiently simple distribution class $\mathcal{P}$ (specifically, a low-degree piecewise polynomial distribution with not too many pieces), and if (2) we can find the best fit in the class $\mathcal{P}$ to the empirical distribution in $\mathcal{A}_k$ distance, then one can show that this best fit distribution must be close to the ground truth distribution.
So, in other words, for nice distribution classes, if we let $f_n$ be the empirical distribution after $n$ samples, it suffices to solve the following optimization problem:
\begin{equation}
    \label{eq:ak-objective}
    \hat{f}_n \triangleq \argmin_{f_\theta \in \mathcal{P}} \dak (f_\theta,f_n) \; .
\end{equation}
They then demonstrate that this problem can be solved in nearly-linear time, via a two-layer algorithm: first, a method which allows them to find a good polynomial approximation to the empirical distribution in a fixed interval in $\operatorname{d}_{\mathcal{A}_1}$ distance, and second, a ``greedy merging'' algorithm which, given such a method, can find good breakpoints for a piecewise polynomial fit, in nearly-linear time.

The objective~\Cref{eq:ak-objective} is a special case of a well-studied class of estimators called \emph{minimum-distance estimators}~\citep{yatracos1985rates} that we alluded to above, which are a classic approach for achieving good guarantees for estimation in total variation distance.
Unfortunately, these estimators are very tailored to total variation distance, and \emph{a priori} seem very hard to generalize to other distances, including Hellinger distance.

\paragraph{A minimum-distance-style estimator for Hellinger distance.} Our main conceptual contribution is the development of a new minimum-distance-esque estimator that allows us to achieve nearly optimal statistical guarantees for learning in Hellinger distance.
We show that the ``correct'' analog of the $\mathcal{A}_k$ distance for learning in Hellinger distance is the following distance:
\begin{equation}\label{eq:hk}
\mathcal{H}_1 (f,g) \triangleq \sup_{I = [a, b]} \frac{1}{2} \left( \sqrt{\Pr_f\left[x \in I \right]} - \sqrt{\Pr_g\left[ x \in I \right]} \right)^2 \; .
\end{equation}
Interestingly, note that this distance is really the analog of the $\mathcal{A}_1$ distance. For total variation distance, it was important that we took $k \geq 1$, whereas for Hellinger, it suffices to consider $k = 1$.
Our main statistical result is that if we take the minimum-distance estimator for the empirical distribution from the class of piecewise polynomial distributions, i.e.
\begin{equation}
    \label{eq:h1-objective}
    \hat{f}_n \triangleq \argmin_{f_\theta \in \mathcal{P}} \mathcal{H}_1 (f_\theta,f_n) \; .
\end{equation}
then this achieves optimal statistical guarantees for learning the original distribution in Hellinger distance (see~\Cref{cor:poly-final-h1-opt}), for appropriate choices of parameters for $\mathcal{P}$.
By leveraging this result, we obtain new, nearly statistically optimal estimators for a number of interesting distribution families in Hellinger distance, including the aforementioned classes of mixtures of Gaussians and mixtures of log-concave distributions.

From a technical point of view, the key idea to proving this is to demonstrate a new Yatracos-style guarantee for the optimality of a new general minimum-distance estimator for Hellinger distance.
Here, we crucially use a recent reverse data processing inequality of~\cite{pensia2023communication}, which allows us to relate the Hellinger distance between two distributions to the Hellinger distance between thresholded versions of these distributions.

We believe this guarantee may be of independent interest beyond our setting.
The fact that minimum-distance estimators, which seem very specially tailored to total variation distance, can be adapted to yield such tight rates for learning in Hellinger, seems quite surprising, and could also have applications for understanding the landscape of density estimation in Hellinger distance more generally.

\paragraph{A new greedy merging algorithm for $\mathcal{H}_1$}
Armed with this new minimum-distance estimator, we are left with the task of efficiently computing the solution to~\Cref{eq:h1-objective} when $\mathcal{P}$ is the class of $t$-piece degree-$d$ polynomials.
To do so, we demonstrate that the greedy merging framework of~\cite{acharya2017sample} can be adapted to this task.
We will have to alter both layers of the algorithm as described above.
First, we design a novel oracle for fitting a degree-$d$ polynomial to the empirical distribution in $\mathcal{H}_1$.
Then, we show how the merging algorithm can be changed to use this oracle to obtain a good piecewise fit to the overall distribution.

\paragraph{Overview of the paper.} We discuss related work in \Cref{sec:related-work} and preliminaries in \Cref{sec:preliminaries}. In \Cref{sec:min-dist}, we present the classical minimum-distance estimator. In \Cref{sec:reverse-dpi}, we describe a recent result on \textit{reverse data processing inequalities}, which may seem unrelated at first glance. 

In \Cref{sec:min-dist-hel}, we provide our core contribution: we describe a modified minimum-distance estimator and prove general guarantees for learning in Hellinger distance with this technique. A key technical tool is drawing a connection to the recent reverse data processing inequality results. In \Cref{sec:stat-apps}, we apply this general result to recover statistical guarantees in Hellinger distance for classes like mixtures of log-concave densities, mixtures of Gaussians, and mixtures of specified densities. These results are often near-optimal up to logarithmic factors.

In \Cref{sec:log-concave}, we leverage how our new recipe is quite flexible; our recipe is not only statistical, it will enable fast (even near-linear time) algorithms for estimating in Hellinger after only a few additional ideas. We demonstrate how the technique of \cite{acharya2017sample} can be massaged with our new recipe to then learn densities that are well-approximated by piecewise polynomials --- eventually yielding our \Cref{corr:fast-gauss} for mixtures of Gaussians, and \Cref{thm:fast-log-concave} for mixtures of log-concave densities. These are the first near-linear time estimators with near-optimal sample dependencies for these tasks; to our knowledge, we are not even aware of any previous polynomial-time algorithms in the literature using near-optimal samples.

We conclude with a discussion in Appendix~\ref{sec:discussion}.

\subsection{Related work}\label{sec:related-work}
\textit{Density estimation via polynomial approximation.} The works of \cite{chan2013learning,chan2014efficient} demonstrated a minimum-distance estimator methodology where efficient algorithms and statistical guarantees are implied by proving that the densities of interest are well-approximated by piecewise polynomial densities (with a bounded number of pieces and bounded degree); for example, \cite{chan2014efficient} yielded an efficient algorithm for near-optimally learning log-concave mixtures in polynomial time within total variation distance. This spurred further work with this method for density estimation and testing properties of distributions (e.g. \cite{diakonikolas2014testing,diakonikolas2015differentially,diakonikolas2016efficient,acharya2015fast,acharya2017sample,canonne2018testing,chen2020efficiently}). Our work enables the implementation of this method for Hellinger distance (even in near-linear time via the acceleration of \cite{acharya2017sample}).

These minimum-distance estimators are defined with respect to $\mathcal{A}_k$ distance, which enforces that for any union of $k$ disjoint intervals, the counts within these intervals are roughly the same for the estimate and the true distribution, in an $\ell_1$ sense (further work related to this notion includes e.g. \cite{diakonikolas2015optimal,diakonikolas2017near,diakonikolas2019testing,diakonikolas2024testing,gerber2025density}).

\textit{Applications of reverse data processing inequalities for estimation. } In \Cref{sec:reverse-dpi}, we will discuss the reverse data processing inequality results of \cite{bhatt2021information,pensia2023communication}.  These results naturally have applications in communication-constrained estimation, yet they have also found some applications in estimation without any communication or memory constraints. The work of \cite{compton2025attainability} aims to design an adaptive estimator of the mean, where one estimator is simultaneously near-optimal for a wide range of symmetric, unimodal distributions. One key feature is that they want their estimate of the mean $\hat{\mu}$ to satisfy that the empirical number of samples within $[\hat{\mu}+a,\hat{\mu}+b]$ is close to the number of samples in the reflected interval $[\hat{\mu}-b,\hat{\mu}-a]$, for all $0 \le a \le b$; since the distribution is symmetric, this would hold if $\hat{\mu} = \mu$. Unlike $\mathcal{A}_k$ distance, their notion of closeness is in the Hellinger sense instead of the $\ell_1$ sense, and the analysis of their estimator crucially leverages a reverse data processing inequality and second-order uniform convergence guarantees. This is quite similar to a notion of $\mathcal{H}_k$ distance we define later in this paper. The work of \cite{blanc2025instance} also uses statistics of intervals in combination with a reverse data processing inequality to design an instance-optimal algorithm for uniformity testing. 

\textit{Learning mixtures of Gaussians.} There is a rich body of work for learning mixtures of Gaussians. Focusing on the univariate setting, there are two main lines of work for efficient density estimation with near-optimal sample dependence in Hellinger distance. A line of work on the nonparametric maximum likelihood estimator (NPMLE) works well when the underlying distribution can be viewed as $\pi * N(0,1)$ for some mixing distribution $\pi$; this will even yield a near-linear time algorithm when all mixture components have similar variances \citep{saha2020nonparametric,polyanskiy2020self,polyanskiy2025nonparametric}, yet without assumptions on the variances this method does not clearly yield efficient algorithms (in fact, the NPMLE is not well-defined without variance assumptions). Method of moment estimators (e.g. see \cite{WuYang}) also focus on this setting with well-conditioned variances. Trying to reduce the arbitrary variances case to the well-conditioned variances case via localization techniques (like that of \cite{liu2022robust}) appears to incur suboptimal sample complexity with known techniques.
In contrast, for total variation distance, the piecewise polynomial method of \cite{acharya2017sample} yields a near-linear time algorithm. We will obtain a near-linear time algorithm for Hellinger distance, which to our knowledge is even the first polynomial time algorithm with near-optimal sample dependence.

\textit{Log-concave density estimation.} Estimating multivariate log-concave densities with near-optimal samples in Hellinger has been a large problem of interest (e.g. \cite{seregin2010nonparametric,doss2016global,KimSamworth,diakonikolas2017learning,carpenter2018near}), and the optimal rate for $d$ dimensions has only recently been solved by \cite{kur2019optimality}. These works primarily study the NPMLE; this can be computed in polynomial time \citep{axelrod2019polynomial} for a single log-concave density, but is not well-defined for mixtures of log-concave densities. Please see the survey of \cite{surveySamworth} for a more general treatment on the log-concave density estimation literature. Most work on learning log-concave mixtures has a more methodological perspective than a theoretical focus (e.g. \cite{cule2010maximum,hu2016maximum}) --- other than for total variation where \cite{acharya2017sample} yields a near-linear time algorithm with near-optimal sample complexity. To our knowledge, no prior work gives a polynomial-time algorithm with near-optimal sample complexity for log-concave mixtures in Hellinger. (It seems plausible this is achievable with alternative methods like wavelet techniques, although we have not tried executing this plan ourselves, and did not find any references that do so.) We will demonstrate a near-linear time algorithm with near-optimal sample complexity.

\textit{Le Cam--Birg\'e covering arguments. } A foundational approach for proving statistical guarantees involves leveraging bounds for the covering number to obtain upper bounds \citep{lecam1973convergence,birge2006model}. Combining covering number bounds with appropriate pairwise composite hypothesis testers (e.g. \cite{birge2013robust,suresh2021robust}) yields optimal Hellinger density estimation guarantees for a large collection of problems; for example, see Chapter 32.2 of \cite{polyanskiy2025information} for a more comprehensive overview. While this technique often pins down optimal rates, it does not imply a polynomial time algorithm. We also remark that in some settings of interest in this paper, the local covering numbers are not bounded, yet we may obtain desirable guarantees. For example, when the class of distributions is the set of all mixtures of at most two univariate Gaussians, then there is no finite-sized cover for the ball around the density which is just one Gaussian (the location of the second Gaussian could be anywhere).

\textit{$\rho$-estimation.} Work introducing $\rho$-estimation \citep{baraud2017new,baraud2016rho,baraud2018rho} builds upon the Le Cam--Birg\'e style of argument, by providing a very general estimator that handles limitations such as the previous example discussed. Among other insights, their work leverages how the required pairwise testing is controlled by a nice function that is bounded, Lipschitz, and monotonic in the likelihood ratio of the two densities. Accordingly, they may obtain desirable estimation guarantees when this function concentrates in ways that hold from empirical processes techniques. Notably for us, these conditions hold in terms of the VC dimension of the same concept class we will leverage in our work, and hence $\rho$-estimation will enable the same statistical guarantees (up to logarithmic factors) as all our statistical corollaries. With this in mind, the key novelties of our work are: (i) these statistical guarantees are also attainable by minimum-distance estimators, and (ii) $\rho$-estimation does not naively run in polynomial time, while our techniques may be more conducive to designing faster algorithms as we will later demonstrate. Further developments relating to $\rho$-estimation emphasizing computational complexity include e.g.  \cite{sart2014estimation,sart2016robust,sart2021estimating}.

\textit{Estimation in KL divergence. } Since KL divergence upper bounds squared Hellinger distance, estimators with small KL error would be desirable in the settings we study; however, we will informally describe why it is impossible to get such general guarantees in all our settings. The technique of \cite{yang1999information} gives an expected-error bound in terms of the KL covering number, but the covering number can be too large in our settings (see the previous two-Gaussian example). The same two-Gaussian example also illustrates how it is impossible to achieve small KL error for classes like Gaussian mixtures in the $\operatorname{KL}(f,\hat{f}_n)$ direction. For impossibility in the $\operatorname{KL}(\hat{f}_n,f)$ direction, consider a mixture of two densities that are uniform over an interval of the real line (both densities are log-concave), such that the resulting mixture is uniform over $[0,1]$, except it is not supported in some small unknown interval inside (say, of width $1/n^3$). An estimator will incur infinite KL error if it puts support on this interval, yet it does not know the location of this interval. Hence, an estimator will typically incur large KL error because it will either output a density with support on this interval (incurring infinite error), or a density that is not supported on a large fraction of $[0,1]$ (which also incurs large KL error). This implies a lower bound for learning mixtures of log-concave densities in KL, yet desirable guarantees are possible in Hellinger.

\subsection{Preliminaries}\label{sec:preliminaries}
For simplicity, all results in this paper assume the target distribution has no atoms. We also assume $n \ge 2$. $\log$ is base 2. The notation $a \lesssim b$ denotes that there is some absolute constant $C > 0$ where $a \le Cb$. A distribution is log-concave if the logarithm of its density is concave. We use $\mathcal{LC}_k$ to denote the class of mixtures of $k$ univariate log-concave densities, and similarly $\mathcal{G}_k$ for Gaussian mixtures. $f_n$ denotes the empirical distribution from $n$ samples. We use $\dhsq(p_\mathcal{I},q_\mathcal{I})$ for restricting the integral in squared Hellinger distance to just an interval $\mathcal{I}$.

\textbf{Uniform convergence. } It is typical to use uniform convergence guarantees for a concept class $\mathcal{F}$ that show with high probability $\sup_{f \in \mathcal{F}} |\sum_{i=1}^n (f(X_i) - \Ex[f(X_i)]) | \lesssim \sqrt{n}$. We will leverage slightly less common results that enable tighter concentration when $\Ex[f]$ is small:
\begin{corollary}[implied by Theorem 5.1 of \cite{boucheron2005theory}; see Appendix~\ref{pf:sqrt-converge}]
    \label{cor:sqrt-converge}
Let $X_1, \ldots, X_n$ be i.i.d. random variables, and assume that the class $\mathcal{F}$ of $\{0, 1\}$-valued functions has VC dimension $d$. Then for any $\delta \in (0, 1)$, with probability at least $1 - \delta$, all $f \in \mathcal{F}$ satisfy
\begin{equation*}
\left|\sqrt{\sum_{i = 1}^n f(X_i)} - \sqrt{ \sum_{i=1}^n\Ex[ f(X_i)]}\right| \le 8 \sqrt{d \ln(2n + 1) + \ln \frac{8}{\delta}}.
\end{equation*}
\end{corollary}

\subsection{Minimum-distance estimators}\label{sec:min-dist}
The classical technique of minimum-distance estimators \citep{yatracos1985rates} appears in textbooks such as in Chapters 6 and 8 of \cite{devroye2001combinatorial} and Chapter 32.3 of \cite{polyanskiy2025information}; we will provide an overview in the style of \cite{devroye2001combinatorial}. Let $\mathcal{F}$ be a class of densities, and suppose $\mathcal{A}$ is some concept class. We define a notion of distance quantifying how much two densities $p,q$ differ with respect to $\mathcal{A}$:

\begin{equation}\label{eq:tv-est-def}
    \da(p,q) \triangleq \sup_{A \in \mathcal{A}} \left| \Pr_{p}[x \in A] - \Pr_{q}[x \in A] \right|
\end{equation}

The intuition of the minimum-distance estimator is that we will choose an estimate $\hat{f}_n$ that has small distance from the empirical samples $f_n$ according to $\mathcal{A}$:

\begin{equation}
    \hat{f}_n \triangleq \argmin_{f_\theta \in \mathcal{F}} \da(f_\theta,f_n)
\end{equation}

(Throughout this paper, if the minimizer does not exist, the same arguments will allow for an approximate minimizer.) For some arbitrary concept class $\mathcal{A}$, the minimum-distance estimator might not be a good estimate. However, it is fruitful to use this estimator when $\mathcal{A}$ is the \textit{Yatracos class} for $\mathcal{F}$, meaning that it contains all sets describing where one density in $\mathcal{F}$ exceeds another:

\begin{equation*}
    \mathcal{A} \triangleq \bigl\{ \{x : f_\theta(x) > f_{\theta'}(x) \} :  \, f_\theta, f_{\theta'} \in \mathcal{F} \bigr\}.
\end{equation*}

In this case, there is a clean, good bound on the performance of $\hat{f}_n$ in terms of the VC dimension of the Yatracos class:

\begin{lemma}[Minimum-distance estimate, e.g. described in \cite{devroye2001combinatorial}]\label{lemma:tv-dist}
    Given $n$ samples from a distribution $f$, let $\hat{f}_n$ be the minimum-distance estimator with respect to the Yatracos class $\mathcal{A}$ for a set of densities $\mathcal{F}$ ($f$ need not be in $\mathcal{F}$). Then, 

    \begin{equation}
        \Ex[\tv(f,\hat{f}_n)] \le 3 \inf_{f_\theta \in \mathcal{F}} \tv(f,f_\theta) + O \left(\sqrt{\frac{\operatorname{VC}\left(\mathcal{A}\right)}{n}} \right)
    \end{equation}
\end{lemma}

This is a technique which often yields near-optimal guarantees when aiming to learn a set of densities whose Yatracos class has small VC dimension. Similarly, this is fruitful if the densities are well-approximated by another set of densities with amenable Yatracos class. For a particularly notable example, the Yatracos class for log-concave densities has unbounded VC dimension, yet they are well-approximated by piecewise linear densities, and this enables a near-optimal guarantee for learning in total variation distance \citep{chan2014efficient}.

Upon first glance, this technique seems quite tailored to total variation distance. An immediate obstacle to using such techniques for learning with Hellinger distance, is the crucial step leveraging how Yatracos classes capture TV distance. This followed from how the total variation distance between any pair of distributions $p,q$ has its distance captured by counting the number of samples in a particular region (say, the region where $p(x)>q(x)$):

\begin{equation*}
    \tv(p,q) = \left|\Pr_{p}[p(x) > q(x)] - \Pr_q[p(x) > q(x)] \right|
\end{equation*}

Since the Yatracos class contains these regions of interest, the minimum-distance estimate is close in total variation distance. On the other hand, it is not obviously clear for Hellinger that we can always near-optimally distinguish between $p,q$ by counting the number of samples in a particular region --- until recently.

\subsection{Reverse data processing inequalities}\label{sec:reverse-dpi}
We aim to leverage recent observations on \textit{reverse data processing inequalities} to overcome this obstacle. Roughly, independent results of \cite{bhatt2021information,pensia2023communication} show that counting the number of samples in certain likelihood-ratio regions can nearly preserve Hellinger distance. Consider two distributions $p,q$, and let $f_\tau$ denote the indicator of their likelihood ratio exceeding a threshold:

\begin{equation*}
    f_\tau (x) \triangleq \mathds{1}\left[\frac{p(x)}{q(x)} \ge \tau\right]
\end{equation*}

Let $f_\tau(p)$ denote the distribution that samples from $p$, applies $f_\tau$, and only observes its output. Meaning, $f_\tau(p)$ is equivalently $\operatorname{Bernoulli}\left(\Pr_p\left[\frac{p(x)}{q(x)} \ge \tau\right]\right)$. The surprising reverse data processing inequality shows how there always exists some $f_\tau$ that nearly preserves Hellinger distance. We rephrase the form of the result given in \cite{pensia2023communication} as:

\begin{theorem}[Weaker version of Corollary 3.4 of \cite{pensia2023communication}; preservation of Hellinger distance]\label{thm:rdpi}
    For any $p,q$, there exists a $\tau$ such that the following holds:
    \begin{equation}
    \dhsq(f_\tau(p), f_\tau(q)) \ge \frac{\dhsq(p,q)}{50\log(4/\dhsq(p,q))}
    \end{equation}
\end{theorem}
Roughly, this means that if two distributions required $k$ samples to well-distinguish, then they are still well-distinguishable from only counting the number of samples in some region with $O(k \log k)$ samples. For a self-contained presentation, we give a two-page proof in Appendix~\ref{sec:pensia} (the constants in \Cref{thm:rdpi} reflect the constants given by our presentation). Our presentation is slightly shorter than that of \cite{pensia2023communication} since we do not require as general a result as theirs, but we emphasize that no technique is notably different from their method. We remark that a generalization of this reverse data-processing inequality is proven in \cite{kazemi2025sample} for ``TV-like $f$-divergences.''

\section{Minimum-distance estimates for Hellinger}\label{sec:min-dist-hel}
We now present a simple modified recipe in order to get Hellinger guarantees from a minimum-distance estimator. Our approach will leverage a combination of second-order uniform convergence guarantees (where deviations are of the order $\sqrt{\Ex[f]/n}$ instead of $\sqrt{1/n}$), and this reverse data processing inequality (this combination has also been fruitful recently in a different setting studying adaptive mean estimation \citep{compton2025attainability}; see \Cref{sec:related-work} for a discussion). 

The approach will be quite similar to the classical minimum-distance estimator proof, with some modifications. Let $\mathcal{F}$ be a class of densities, and suppose $\mathcal{H}$ is some concept class. We define a modified notion of distance quantifying how much two densities differ with respect to $\mathcal{H}$:
\begin{equation}
    \hdist(p,q) \triangleq \sup_{H \in \mathcal{H}} \frac{1}{2} \left(\sqrt{\Pr_p[x \in H]} - \sqrt{\Pr_q[x \in H]} \right)^2
\end{equation}
This distance is more well-equipped for Hellinger than the analogous estimator in \Cref{eq:tv-est-def}. We will choose an estimate $\hat{f}_n$ that looks as it should according to $\mathcal{H}$:
\begin{equation}
    \hat{f}_n \triangleq \argmin_{f_\theta \in \mathcal{F}} \hdist(f_\theta, f_n)
\end{equation}
Instead of the Yatracos class, we consider $\mathcal{H}$ to be what we call a \textit{ratio class}; this class contains all sets describing where the ratio between two densities in $\mathcal{F}$ exceeds some value:
\begin{equation}
    \mathcal{H} \triangleq \left\{ \left\{ x : \frac{f_\theta(x)}{f_{\theta'}(x)} \ge \tau \right\} :  f_\theta,f_{\theta'} \in \mathcal{F}, \tau \in \mathbb{R} \right\}.
\end{equation}
Note how this generalizes the Yatracos class (where $\tau = 1$). Our main contribution shows that if $\mathcal{H}$ is the ratio class, then we conclude a similar bound in Hellinger distance for the performance of $\hat{f}_n$:
\begin{theorem}\label{thm:hel-min-dist}
    Given $n$ samples from a distribution $f$, let $\hat{f}_n$ be the minimum-distance estimator with respect to the ratio class $\mathcal{H}$ for a set of densities $\mathcal{F}$ ($f$ need not be in $\mathcal{F}$). Suppose the VC dimension of $\mathcal{H}$ is $d\ge 1$. Then with probability $1-\delta$,
    \begin{equation}
        \dhsq(f,\hat{f}_n) \lesssim \inf_{f_\theta \in \mathcal{F}}\dhsq(f,f_\theta) \log(2/\dhsq(f,f_\theta)) + \frac{(d \log(n) + \log \frac{2}{\delta}) \log(n)}{n} 
    \end{equation}
\end{theorem}
\begin{proof}
    Our proof will follow similarly to the textbook proof of \Cref{lemma:tv-dist}, with the correct tools now in-hand. Conveniently, $\hdist$ is nicely behaved; it immediately holds that $\hdist(f,g) \le \dhsq(f,g)$, and it satisfies an approximate triangle inequality $\hdist(a,c) \le 2 \cdot (\hdist(a,b)+\hdist(b,c))$. We track (unoptimized) constants for this statistical result in case it is helpful for future downstream applications. Starting off, by the approximate triangle inequality, for any $f_\theta \in \mathcal{F}$ it holds that:
    \begin{equation}
        \dhsq(f,\hat{f}_n) \le 2 \cdot (\dhsq(f,f_\theta) + \dhsq(f_\theta,\hat{f}_n)) \label{step:hel-begin}
    \end{equation}

    We will bound $\dhsq(f_\theta,\hat{f}_n)$ in terms of $\hdist$ by using the reverse data processing inequality (\Cref{thm:rdpi}). From the proof of \Cref{thm:rdpi}, we know for any $f_\theta,f_{\theta'} \in \mathcal{F}$ that $\hdist(f_\theta,f_{\theta'}) \ge \frac{\dhsq(f_\theta,f_{\theta'})}{50 \log(4/\dhsq(f_\theta,f_{\theta'}))}$ (the statement implies a bound worse by a factor of two, but the proof shows the maximum of the two Hellinger summands is at least this quantity, which gives the lower bound for $\hdist$). We want to use this to upper bound $\dhsq(f_\theta,\hat{f}_n)$ in terms of $\hdist(f_\theta,\hat{f}_n)$. Generally, if $0 \le x,y \le 1$ and $x \ge \frac{y}{50 \log(4/y)}$, then we can conclude
    \begin{equation}\label{eq:move-log}
       50x\log(e/x)\ge \frac{y}{\log(4/y)} \cdot \log\left(\frac{50e \log(4/y)}{y}\right) \ge y
    \end{equation}
    since $50x\log(e/x)$ is non-decreasing in terms of $x \in [0,1]$.
    Using this, we continue:
    \begin{align*}
        &2 \cdot (\dhsq(f,f_\theta) + \dhsq(f_\theta,\hat{f}_n)) \le 2\dhsq(f,f_\theta) + 100\hdist(f_\theta,\hat{f}_n) \log(e/\hdist(f_\theta,\hat{f}_n)) \intertext{Let $z(x) \triangleq x \log(e/x)$. Since $z$ is non-decreasing in $[0,1]$, is concave, and $z(0)=0$, then $z(x) \le 2(z(a)+z(b))$ if $0 \le x,a,b \le 1$ and $x \le 2(a+b)$:}
        &\le 2\dhsq(f,f_\theta) + 200\hdist(f,f_\theta) \log(e/\hdist(f,f_\theta)) + 200\hdist(f,\hat{f}_n) \log(e/\hdist(f,\hat{f}_n)) \\
        &\le 202\dhsq(f,f_\theta) \log(e/\dhsq(f,f_\theta)) + 400\hdist(f,f_n) \log(e/\hdist(f,f_n)) + 400\hdist(f_n,\hat{f}_n) \log(e/\hdist(f_n,\hat{f}_n)) \intertext{By the definition of $\hat{f}_n$, we know $\hdist(f_n, \hat{f}_n) \le \hdist(f_n,f_\theta)$:}
        &\le 202\dhsq(f,f_\theta) \log(e/\dhsq(f,f_\theta)) + 400\hdist(f,f_n) \log(e/\hdist(f,f_n)) +400\hdist(f_n,f_\theta) \log(e/\hdist(f_n,f_\theta))\\
        &\le 202\dhsq(f,f_\theta) \log(e/\dhsq(f,f_\theta)) + 1200\hdist(f,f_n) \log(e/\hdist(f,f_n)) +800\hdist(f,f_\theta) \log(e/\hdist(f,f_\theta))\\
        &\le 1002\dhsq(f,f_\theta) \log(e/\dhsq(f,f_\theta)) + 1200\hdist(f,f_n) \log(e/\hdist(f,f_n)) \intertext{
    To finish, we will use the second-order uniform convergence bound of \Cref{cor:sqrt-converge}:}
        & \le 1002\dhsq(f,f_\theta) \log(e/\dhsq(f,f_\theta)) + 1200 \cdot \left(32(d \ln(2n + 1) + \ln \frac{8}{\delta})/n \right) \log\left(\frac{en}{32(d \ln(2n + 1) + \ln \frac{8}{\delta}) }\right)\\
        &\le 1002\dhsq(f,f_\theta) \log(e/\dhsq(f,f_\theta)) + \frac{38400(d \ln(2n + 1) + \ln \frac{8}{\delta}) \log(n)}{n} %
    \end{align*}
\end{proof}
This yields a near-identical recipe for density estimation in Hellinger with minimum-distance estimators; to learn a class of densities, merely show that it is well-approximated by densities whose ratio class has small VC dimension.

\textbf{Sharpening the misspecification term. } The guarantee from \Cref{thm:hel-min-dist} has an extra logarithmic factor in the misspecification term that does not usually appear in this style of argument. We remark that it is possible to sharpen this guarantee to remove the logarithmic factor, if we use a modified version of $\hdist$ that corresponds to quantizing into $\Theta(\log(n))$ states instead of 2 states. The main idea is that either: (i) $\dhsq(f,f_\theta) \le \frac{1}{n}$, so the misspecification term can be absorbed by the other term, or (ii) the more complete version of Corollary 3.4 from \cite{pensia2023communication} will only lose a constant factor of Hellinger distance when quantized to $\Theta(\log(n))$ states. We state the sharper result here, but defer the proof to Appendix~\ref{sec:sharper-misspec}, since our efficient algorithms are based on the estimator of \Cref{thm:hel-min-dist}.

\begin{theorem}\label{thm:sharper-hel-min-dist}
    Given $n$ samples from a distribution $f$, let $\hat{f}_n$ be a modified minimum-distance estimator with respect to the ratio class $\mathcal{H}$ for a set of densities $\mathcal{F}$ ($f$ need not be in $\mathcal{F}$). Suppose the VC dimension of $\mathcal{H}$ is $d\ge 1$. Then with probability $1-\delta$,
    \begin{equation}
        \dhsq(f,\hat{f}_n) \lesssim \inf_{f_\theta \in \mathcal{F}}\dhsq(f,f_\theta) + \frac{(d \log(n) + \log \frac{2}{\delta}) \log(n)}{n} 
    \end{equation}
\end{theorem}

\subsection{Applications}\label{sec:stat-apps}
We can now immediately adapt arguments from minimum-distance estimators in total variation. Below are three examples where we recover (previously-known) bounds in Hellinger distance (details in Appendix~\ref{app:stat-corr}); all are tight up to logarithmic factors. We emphasize that these statistical corollaries are all previously attainable via $\rho$-estimation (up to logarithmic factors, see discussion in \Cref{sec:related-work}).

\begin{corollary}\label{cor:stat-log}
    Given $n$ samples from some distribution $f$, then a minimum-distance estimator $\hat{f}_n$, with probability $1-\delta$, achieves error
    \begin{equation*}
        \dhsq(f,\hat{f}_n) \lesssim \inf_{f_\theta \in \mathcal{LC}_k} \dhsq(f,f_\theta) + \frac{k^{4/5} \log^{4/5}(n)}{n^{4/5}} + \frac{\log(2/\delta) \cdot \log(n)}{n}.
    \end{equation*}
\end{corollary}

\begin{corollary}\label{cor:stat-gaussian}
    Suppose $\mathcal{G}_k$ denotes the class of mixtures of $k$ Gaussians. Given $n$ samples from some distribution $f$, then a minimum-distance estimator $\hat{f}_n$, with probability $1-\delta$, achieves error
    \begin{equation*}
        \dhsq(f,\hat{f}_n) \lesssim \inf_{f_\theta \in \mathcal{G}_k} \dhsq(f,f_\theta)+ \frac{(k\log^2(n) + \log \frac{2}{\delta}) \log(n)}{n}.
    \end{equation*}
\end{corollary}

\begin{corollary}\label{cor:stat-mix}
    Let $\mathcal{F}$ denote the set of convex combinations of $k$ fixed densities $p_1,\dots,p_k$. Given $n$ samples from some distribution $f$, then a minimum-distance estimator $\hat{f}_n$, with probability $1-\delta$, achieves error
    \begin{equation*}
        \dhsq(f,\hat{f}_n) \lesssim \inf_{f_\theta \in \mathcal{F}} \dhsq(f,f_\theta) + \frac{(k \log(n) + \log \frac{2}{\delta}) \log(n)}{n}.
    \end{equation*}
\end{corollary}

\section{Learning densities approximated by piecewise polynomials}\label{sec:log-concave}
In this section, we discuss how our recipe also enables fast algorithms for learning in Hellinger. We will learn densities that are well-approximated by piecewise polynomials; notably, this enables sample-efficient learning for mixtures of Gaussians or log-concave densities in near-linear time.

\cite{acharya2017sample} provides a near-linear time algorithm for learning piecewise polynomial densities in total variation with near-optimal sample complexity. Their work employs a minimum-distance estimate approach, where they prove their densities of interest are well-approximated by piecewise polynomial functions, they observe that the Yatracos set (in this case, the union of disjoint intervals) has bounded VC dimension, and then they use a greedy partitioning approach to optimize the minimum-distance estimate in near-linear time. This approach is fruitfully applied to many classes of densities that are well-approximated by piecewise polynomials of bounded degree.

In our work, we will show that a similar program can be attained for Hellinger distance by employing ideas similar to \Cref{sec:min-dist-hel} and \cite{acharya2017sample}, with some additional observations:

\begin{corollary}\label{cor:poly-final-h1-opt}
Consider an interval $J=[a,b]$ and $n$ samples $f_n$ where $a \le x_1 \le \dots \le x_n \le b$. Consider $\mathcal{F}_d^t$: the class of $t$-piece degree-$d$ polynomials over $J$. There exists an algorithm that runs in time $O \left( (d^3 \log \log n + nd^2 + d^{\omega + 2}) \log^3 (n) \right)$ and outputs an estimate $\hat{f}_n \in \mathcal{F}_d^{4t}$ where with probability $1-\delta$,
\begin{equation*}
    \dhsq(f,\hat{f}_n) \lesssim  \inf_{f_\theta \in \mathcal{F}_d^t} d\dhsq(f,f_\theta) \log(2t / \dhsq(f,f_\theta)) + \frac{dt \log(n) \cdot (\log(n) + \log(2/\delta))}{n}.
\end{equation*}
\end{corollary}

This implies our \Cref{corr:fast-gauss} for mixtures of Gaussians. Our \Cref{thm:fast-log-concave} for log-concave densities uses a slightly different result, where we improve a logarithmic factor by proving a tighter reverse data processing inequality for a restricted class of piecewise-linear densities that well-approximate log-concave densities.

\subsection{Technique overview} 
We now outline our technique for adapting the method of \cite{acharya2017sample} to Hellinger distance. We provide much more detail in \Cref{app:log-concave,app:piecewise-poly}, where we choose to carefully outline the required modifications of their method, instead of essentially rewriting all their corresponding 30 pages, since our modifications only require changing some cleanly-describable subroutines.

\textbf{Component 1: proving bounded empirical $\mathcal{H}_1$ distance implies small Hellinger distance. } Their approach uses a classical minimum-distance estimator argument; after establishing that a log-concave density may be well-approximated by a piecewise linear density, it suffices to output a $\hat{f}_n$ with small $\mathcal{A}_k$ distance (see \Cref{eq:ak}) from $f_n$. This roughly quantifies a distance between two densities with respect to what can be observed from the count of samples within $k$ disjoint intervals. 

In \Cref{eq:hk}, we define a notion of $\mathcal{H}_1$ distance. With an argument similar to \Cref{thm:hel-min-dist}, we will obtain a bound for $\dhsq(f,\hat{f}_n)$ in terms of $\mathcal{H}_1(\hat{f}_n,f_n)$. Naively this bound would need to be in terms of some $\mathcal{H}_k$ distance, but we modify the argument such that instead a bound on $\mathcal{H}_1$ suffices; we note that such a modification is not possible for the $\mathcal{A}_k$-based procedure in \cite{acharya2017sample}. We now turn towards efficiently optimizing $\mathcal{H}_1(f_n,\hat{f}_n)$.

\textbf{Component 2: optimizing $\mathcal{H}_1$ distance over $t$-piece densities via greedy merging, assuming an oracle for optimizing over 1-piece densities.} The method of \cite{acharya2017sample} gives an algorithmic procedure, where optimizing over $t$-piece densities may be reduced to the task of optimizing over 1-piece densities. We will use this same procedure. Afterwards, we apply an additional post-processing step that rescales each piece of $\hat{f}_n$ to have the same mass as $f_n$ within each piece, as this is crucial for our argument that allowed us to focus on only $\mathcal{H}_1$ instead of $\mathcal{H}_k$.

\textbf{Component 3: designing an oracle for optimizing $\mathcal{H}_1$ distance over 1-piece densities.} We now want to find a degree-$d$ polynomial that approximately minimizes the $\mathcal{H}_1$ distance with the empirical distribution within a given interval $J$. The core technical obstacle for adopting their technique reduces to the following: for some given polynomial $p$, find some interval $I \subseteq J$ that witnesses most of the distance $\mathcal{H}_1(p,f_n)$. In the analogous subroutine for \cite{acharya2017sample}, their task reduced to a previously-studied question of \cite{csuros2004maximum}. For our work, finding an interval that witnesses $\mathcal{H}_1$ distance will require a new approximation algorithm.

We first reduce our task to the following optimization problem over an array of $\pm 1$ values, where we want a constant-factor approximation in linear time: 
\begin{equation*}
\max_{l\le r \text{ where } l,r\in \{0,\dots,|A|-1\}} \frac{1}{2} \left(\sqrt{\text{(\# of -1's in $A_l,\dots,A_r$)} /n} - \sqrt{\text{(\# of +1's in $A_l,\dots,A_r$)} /n} \right)^2
\end{equation*}
Note how this is trivially solved exactly in quadratic time. Oversimplifying, we roughly approximate the objective by the quantity $\frac{| (\text{\# of -1's in }A_{l},\dots,A_{r}) - (\text{\# of +1's in }A_{l},\dots,A_{r}) |^2}{n \cdot (r-l+1)}$. The numerator is easy to optimize in $O(1)$ time within some interval via prefix sums and a range minimum query data structure. We prove that if we optimize for this numerator within a carefully chosen set of intervals, then one of the optimizer intervals witnesses the $\mathcal{H}_1$ distance sufficiently well.

When combined, these components yield a near-linear time algorithm for learning mixtures of Gaussians and log-concave densities with near-optimal samples (\Cref{corr:fast-gauss,thm:fast-log-concave}).

\acks{We thank Henry WJ Reeve for a helpful discussion relating to the sharpened misspecification term in \Cref{thm:sharper-hel-min-dist}. S.C. is supported by the NDSEG Fellowship Program, Tselil Schramm’s
NSF CAREER Grant no. 2143246, and Gregory Valiant’s Simons Foundation Investigator Award and NSF award
AF-2341890.}

\bibliography{ref}

\appendix

\section{Discussion}\label{sec:discussion}
In this work, we developed fast, sample-efficient algorithms for density estimation in Hellinger. Our primary applications focused on mixtures of log-concave or Gaussian densities; classical total variation guarantees for minimum-distance estimation seem ripe for adoption with our recipe, and we believe this can achieve nearly-optimal results for most, if not all, distribution families considered in prior work such as~\cite{acharya2017sample} (once you conclude the corresponding piecewise approximation guarantee). We find it conceptually interesting that \cite{acharya2017sample} required $\mathcal{A}_k$ distance involving unions of $k$ intervals, yet by looking at $\mathcal{H}_1$ distance with only one interval, we are able to recover their total variation results.

One practical consideration for using a density estimation method is whether it requires carefully tuning some parameters. For the piecewise polynomial approach, we observe how there are not too many important parameter configurations; since we only need to decide on the number of pieces and the degree, and our analyses work when these parameters are off by a constant factor, then only a polylogarithmic number of configurations are relevant. With this in mind, you could in near-linear time try all configurations, then use a standard hypothesis selection method, and obtain an adaptive algorithm that requires no parameter tuning at all. 

We note that there are similarly-named minimum-distance methods that are applied to Hellinger (e.g. \cite{wolfowitz1957minimum,beran1977minimum,donoho1988automatic}; see \cite{jana2025optimal} for a recent work), but these are different from the Yatracos-style estimator we discuss in this work. These estimators work by choosing an estimate in some model class that minimizes some distance with the empirical distribution; this is often some distance like Hellinger or KL. First, the analyses of these estimators follow a different approach that is not the VC-based analysis of Yatracos-style estimators. Second, depending on the setting, these may incur similar issues to the previously-discussed NPMLE (which is one such minimum-distance method, when you choose KL as your distance): for many rich classes this is an intractable optimization problem, or is not even well-defined (like for mixtures of Gaussians with arbitrary variances, the optimizer is not well-defined and must put point masses on empirical samples).

A natural question after reading \Cref{cor:poly-final-h1-opt} is whether there are classes where the degree $d$ of the approximating polynomial must be large, as this would negatively affect both the running time and the squared Hellinger distance error. Interestingly, from personal communication with the authors of \cite{acharya2017sample}, all degree-$d$ pieces may each be replaced with $O(d \log d)$  pieces of degree $O(\log (d/\eps))$, and this will only increase the squared Hellinger distance by $\eps$. This means the dependence on $d$ can be transformed into dependence on $t$, to less dramatically affect the running time and squared Hellinger distance error.

\section{Self-contained proof of the reverse data processing inequality of \cite{pensia2023communication} (\Cref{thm:rdpi})}\label{sec:pensia}

\begin{proof}
    For a self-contained presentation, we will present a proof in this text. Nothing in this presentation is notably different from the technique of \cite{pensia2023communication}, and we will restrict to discrete distributions only for ease of notation. Also, for ease of notation, let $N \triangleq \frac{1}{\dhsq(p,q)}$.

    The proof technique will be to decompose the Hellinger distance into $O(\log(2/\dhsq(p,q)))$ terms, so one must contain a large fraction of the distance, and show that the distance from any term is attainable by a likelihood ratio threshold. To start, by definition of Hellinger distance:
    \begin{equation*}
        \dhsq(p,q) \triangleq \frac{1}{2} \sum_x \left( \sqrt{p(x)} - \sqrt{q(x)} \right)^2 = \frac{1}{2} \sum_{x : p(x) > q(x)} \left( \sqrt{p(x)} - \sqrt{q(x)} \right)^2 + \frac{1}{2} \sum_{x : p(x) < q(x)} \left( \sqrt{p(x)} - \sqrt{q(x)} \right)^2
    \end{equation*}

    Without loss of generality (by symmetry), we will now assume:

    \begin{equation*}
        \sum_{x : p(x) > q(x)} \left( \sqrt{p(x)} - \sqrt{q(x)} \right)^2 \ge \dhsq(p,q)
    \end{equation*}

    We further observe that the summands where $1 < \frac{p(x)}{q(x)} \le 1 + \frac{1}{\sqrt{N}}$ have negligible contribution:

    \begin{align*}
        & \sum_{x: 1 \le \frac{p(x)}{q(x)} \le 1 + \frac{1}{ \sqrt{N}}} \left( \sqrt{p(x)} - \sqrt{q(x)} \right)^2 \\
        & \le \sum_x \left(\sqrt{\left(1 + \frac{1}{ \sqrt{N}}\right)q(x)} - \sqrt{q(x)} \right)^2 \\
        & = \left(\sqrt{1 + \frac{1}{ \sqrt{N}}} - 1 \right)^2 \le \frac{1}{4N} = \frac{\dhsq(p,q)}{4} \\
        & \implies \sum_{x : \frac{p(x)}{q(x)} \ge 1 + \frac{1}{ \sqrt{N}}} \left( \sqrt{p(x)} - \sqrt{q(x)} \right)^2 \ge \frac{1}{2} \dhsq(p,q)
    \end{align*}

    Consider intervals of the form $[1 + \frac{1}{2^{j+1}},1+\frac{1}{2^j})$ for integer $j \ge 0$. Covering $[1 + \frac{1}{\sqrt{N}},2)$ requires $\lceil \log(\sqrt{N}) \rceil$ such intervals. We may additionally cover $[2,\infty)$ by its own interval. Hence, there exists an interval $I^*$, either of the form $[1 + \frac{1}{2^{j+1}},1+\frac{1}{2^j})$ or $[2,\infty)$, where

    \begin{equation*}
        \sum_{x : \frac{p(x)}{q(x)} \in I^*} \left( \sqrt{p(x)} - \sqrt{q(x)} \right)^2 \ge \frac{1}{2 \cdot (\lceil \log(\sqrt{N}) \rceil + 1)} \dhsq(p,q) \ge \frac{1}{2 \cdot \log(4/\dhsq(p,q))} \dhsq(p,q)
    \end{equation*}

    All that remains is to show that for any such $I^*$, there exists an $f_\tau$ with approximately the Hellinger distance of the interval's contribution.

    \textit{Case 1: $I^* = [1+\frac{1}{2^{j+1}},1+\frac{1}{2^j})$. } First, we upper bound the value of the sum from $I^*$:
    \begin{align*}
        & \sum_{x : \frac{p(x)}{q(x)} \in [1+\frac{1}{2^{j+1}},1+\frac{1}{2^j})} \left( \sqrt{p(x)} - \sqrt{q(x)}\right)^2 \\
        & \le \sum_{x : \frac{p(x)}{q(x)} \in [1+\frac{1}{2^{j+1}},1+\frac{1}{2^j})} \left( \sqrt{\left(1+\frac{1}{2^j}\right)q(x)} - \sqrt{q(x)}\right)^2 \\
        & = \Pr_q\left[\frac{p(x)}{q(x)} \in \left[1+\frac{1}{2^{j+1}},1+\frac{1}{2^j}\right)\right] \cdot \left(\sqrt{1 + \frac{1}{2^j}} - 1 \right)^2 \\
        & \le \Pr_q\left[\frac{p(x)}{q(x)} \in \left[1+\frac{1}{2^{j+1}},1+\frac{1}{2^j}\right)\right] \cdot \frac{1}{2^{2(j+1)}}
    \end{align*}
    We may also lower bound the Hellinger distance using $f_\tau$ with $\tau = 1+\frac{1}{2^{j+1}}$:
    \begin{align*}
        \dhsq(f_\tau(p),f_\tau(q)) &\ge \frac{1}{2} \cdot \left( \sqrt{\Pr_p\left[\frac{p(x)}{q(x)} \ge 1+\frac{1}{2^{j+1}}\right]} - \sqrt{\Pr_q\left[\frac{p(x)}{q(x)} \ge 1+\frac{1}{2^{j+1}}\right]} \right)^2\\
        &\ge \frac{1}{2} \cdot \left( \sqrt{\left( 1 + \frac{1}{2^{j+1}} \right) \cdot \Pr_q\left[\frac{p(x)}{q(x)} \ge 1+\frac{1}{2^{j+1}}\right]} - \sqrt{\Pr_q\left[\frac{p(x)}{q(x)} \ge 1+\frac{1}{2^{j+1}}\right]} \right)^2\\
        & = \frac{1}{2} \cdot \Pr_q\left[\frac{p(x)}{q(x)} \ge 1+\frac{1}{2^{j+1}}\right] \cdot \left( \sqrt{1 + \frac{1}{2^{j+1}}} - 1 \right)^2 \\
        & \ge \frac{1}{8} \cdot \frac{1}{2^{2(j+1)}} \cdot \Pr_q\left[\frac{p(x)}{q(x)} \in \left[1+\frac{1}{2^{j+1}},1+\frac{1}{2^j}\right)\right]
    \end{align*}
    Together these imply
    \begin{equation*}
        \dhsq(f_\tau(p),f_\tau(q)) \ge \frac{1}{8} \cdot \sum_{x : \frac{p(x)}{q(x)} \in [1+\frac{1}{2^{j+1}},1+\frac{1}{2^j})} \left( \sqrt{p(x)} - \sqrt{q(x)}\right)^2 \ge \frac{1}{16 \cdot \log(4/\dhsq(p,q))} \dhsq(p,q)
    \end{equation*}
    \textit{Case 2: $I^* = [2,\infty)$. } Let us use $f_\tau$ with $\tau=2$:
    \begin{align*}
        \dhsq(f_\tau(p),f_\tau(q)) &\ge \frac{1}{2} \cdot \left(\sqrt{\Pr_p\left[\frac{p(x)}{q(x)} \ge 2\right]} -\sqrt{\Pr_q\left[\frac{p(x)}{q(x)} \ge 2\right]} \right)^2 \\ 
        &\ge \frac{1}{2} \cdot \left(1 - \sqrt{1/2}\right)^2 \cdot \Pr_p\left[\frac{p(x)}{q(x)} \ge 2\right]\\
        &= \frac{1}{2} \cdot \left(1 - \sqrt{1/2}\right)^2 \cdot \sum_{x : \frac{p(x)}{q(x)} \ge 2} p(x)\\
        & \ge \frac{1}{2} \cdot \left(1 - \sqrt{1/2}\right)^2 \cdot \sum_{x : \frac{p(x)}{q(x)} \ge 2} \left(\sqrt{p(x)} - \sqrt{q(x)}\right)^2\\
        & \ge \frac{\left(1 - \sqrt{1/2}\right)^2}{4 \cdot \log(4/\dhsq(p,q))} \dhsq(p,q) \ge \frac{1}{50 \cdot \log(4/\dhsq(p,q))} \dhsq(p,q) %
    \end{align*}
\end{proof}

\section{Proof of \Cref{cor:sqrt-converge}}\label{pf:sqrt-converge}
\begin{proof}

We first state the result of \cite{boucheron2005theory}, which will imply the desired guarantee after some minor adjustments: 
\begin{lemma}[Normalized uniform convergence; rephrased Theorem 5.1 of \cite{boucheron2005theory}]
\label{lem:vclemma}
Let $X_1, \ldots, X_n$ be i.i.d. random variables, and assume that the class $\mathcal{F}$ of $\{0, 1\}$-valued functions has VC dimension $d$. Then for any $\delta \in (0, 1)$, with probability at least $1 - \delta$, for all $f \in \mathcal{F}$, it holds that both
\begin{equation}\label{eq:unif-first}
\frac{\sum_{i=1}^n \Ex[f(X_i)] - \sum_{i=1}^n f(X_i)}{\sqrt{\sum_{i=1}^n \Ex[f(X_i)]}} \le 2 \sqrt{d \ln(2n + 1) + \ln \frac{8}{\delta}}
\end{equation}
and
\begin{equation}\label{eq:unif-second}
    \frac{\sum_{i=1}^n f(X_i) - \sum_{i=1}^n \Ex[f(X_i)]}{\sqrt{\sum_{i=1}^n f(X_i)}} \le 2 \sqrt{d \ln(2n + 1) + \ln \frac{8}{\delta}}.
\end{equation}
\end{lemma}
Now, we begin the adjustments:
\begin{claim}
    Under the uniform convergence event of \Cref{lem:vclemma}, it holds that for all $f \in \mathcal{F}$:
    \begin{equation*}
        \left| \sum_{i=1}^n \left(f(X_i) - \Ex[f(X_i)]\right)\right| \le 2 \sqrt{\left(d \ln(2n + 1) + \ln \frac{8}{\delta} \right) \cdot \left( \sum_{i=1}^n \Ex[f(X_i)]  \right)} + 4 \left(d \ln(2n + 1) + \ln \frac{8}{\delta}\right)
    \end{equation*}
\end{claim}
\begin{proof}
    The direction $\sum_{i=1}^n \left(\Ex[f(X_i)] - f(X_i)\right)$ holds immediately. For the other direction, we manipulate as follows. For shorthand, let $a \triangleq \sum_{i=1}^n f(X_i)$, $b \triangleq \sum_{i=1}^n \Ex[f(X_i)]$, $\Delta \triangleq a-b$, and $r \triangleq 2 \sqrt{d \ln(2n+1) + \ln \frac{8}{\delta}}$.
    \begin{align*}
        & \frac{a-b}{\sqrt{a}} \le r \implies  \Delta \le r \sqrt{a} \implies \Delta^2 - r^2 a \le 0\\ 
        & \implies \Delta^2 - r^2 (\Delta + b) \le 0 \implies \Delta^2 - r^2 \Delta - r^2 b \le 0 \\
        & \implies \Delta \le \frac{r^2 + \sqrt{r^4 + 4r^2b}}{2} \implies \Delta \le r^2 + r \sqrt{b}  \rtag{root of the quadratic equation} %
    \end{align*}
\end{proof}

We will continue using this shorthand notation of $a,b,\Delta,r$. The proof will follow quickly via the same logic as Claim 2.22 of \cite{compton2025attainability}:

\begin{align*}
     \left| \sqrt{a} - \sqrt{b} \right| &\le \min \left(\frac{2 |a-b|}{\sqrt{b}} , \sqrt{|a-b|}\right) \rtag{for $a,b \ge 0$; first paragraph Claim 2.22 pf. \cite{compton2025attainability}}\\
     & \le \min \left(\frac{2 (r^2 + r \sqrt{b})}{\sqrt{b}} , \sqrt{r^2 + r \sqrt{b}}\right)\rtag{using $|a-b| \le r^2 + r \sqrt{b}$} \\
     & \le \min \left( \frac{2r^2}{\sqrt{b}} + 2r, r + \sqrt{r \sqrt{b}} \right) \intertext{If $b \ge r^2$ then the first term is at most $4r$. If $b \le r^2$ then the second term is at most $2r$. Hence:}
     & \le 4r %
\end{align*}

\end{proof}

\section{Hellinger data processing inequality}
We give a self-contained proof for the following preliminary: \begin{fact}[Hellinger data processing inequality]\label{fact:dpi}
    Consider two points $x,y \in \R_+^d$. Then,
    \begin{equation*}
        \sum_{i=1}^d \left(\sqrt{x_i} - \sqrt{y_i} \right)^2 \ge \left( \sqrt{\sum_{i=1}^d x_i} - \sqrt{\sum_{i=1}^d y_i} \right)^2
    \end{equation*}
\end{fact}
\begin{proof}
    For any non-negative vector $v$, let $\sqrt{v}$ denote the vector where each coordinate $i$ is the square root of the value of $v_i$.
    \begin{align*}
        \sum_{i=1}^d \left( \sqrt{x_i} - \sqrt{y_i} \right)^2 &= \| \sqrt{x} - \sqrt{y} \|^2 \\
        & \ge \| \sqrt{x}\|^2 + \|\sqrt{y}\|^2 - 2 \|\sqrt{x} \| \cdot \| \sqrt{y} \| \\
        & = ( \|\sqrt{x}\| - \| \sqrt{y}\|)^2 = \left(\sqrt{\sum_{i=1}^d x_i} - \sqrt{\sum_{i=1}^d y_i}  \right)^2 %
    \end{align*}
\end{proof}

\section{Proof of \Cref{thm:sharper-hel-min-dist}}\label{sec:sharper-misspec}
\begin{proof}
This proof will follow the same style as \Cref{thm:hel-min-dist}, after we introduce definitions corresponding to quantizing into a larger number of states. In the original distance, $\hdist$ quantizes into one state by choosing the $f,g \in \mathcal{F}$ and threshold $\tau \in \mathbb{R}$ that maximizes distance; our modified distance will choose $f,g \in \mathcal{F}$ and a collection of thresholds $0=\tau_0 < \tau_1 \le \tau_2 \le \dots \le \tau_{D-1} < \tau_D = \infty$ that maximize distance. More concretely:
\begin{equation}
    \hdist^D(p,q) \triangleq \sup_{f,g \in \mathcal{F}, \, \, 0=\tau_0 < \dots  < \tau_D = \infty} \frac{1}{2} \sum_{i=1}^D \left( \sqrt{\Pr_p\left[\frac{f(x)}{g(x)} \in [\tau_{i-1},\tau_i)\right]} - \sqrt{\Pr_q\left[\frac{f(x)}{g(x)} \in [\tau_{i-1},\tau_i)\right]}\right)^2
\end{equation}
Just like $\hdist$, the modified $\hdist^D$ satisfies $\hdist^D(f,g) \le \dhsq(f,g)$, and an approximate triangle inequality $\hdist^D(a,c) \le 2 \cdot (\hdist^D(a,b) + \hdist^D(b,c))$.

We now state the complete version of Corollary 3.4 of \cite{pensia2023communication} in this language:
\begin{theorem}[Rephrased Corollary 3.4 of \cite{pensia2023communication}]\label{thm:stronger-pensia}
    For any $p,q$ and $D \ge 2$, there exists $0=\tau_0 < \tau_1 \le \tau_2 \le \dots \le \tau_{D-1} < \tau_D = \infty$ such that the following holds:
    \begin{equation*}
        \frac{1}{2} \sum_{i=1}^D \left( \sqrt{\Pr_p\left[\frac{p(x)}{q(x)} \in [\tau_{i-1},\tau_i)\right]} - \sqrt{\Pr_q\left[\frac{p(x)}{q(x)} \in [\tau_{i-1},\tau_i)\right]}\right)^2 \ge \frac{\dhsq(p,q)}{1800 \max \left\{1, \frac{\log(4/\dhsq(p,q))}{D} \right\}}
    \end{equation*}
\end{theorem}
The proof of \Cref{thm:stronger-pensia} follows similarly to \Cref{thm:rdpi}. Our modified estimator will choose $D=\lceil \log(4n)\rceil $ and estimate:
\begin{equation}
    \hat{f}_n \triangleq \argmin_{f_\theta \in \mathcal{F}} \hdist^D(f_\theta, f_n)
\end{equation}

The desired error guarantee now follows by the same proof strategy as \Cref{thm:hel-min-dist}:
\begin{align*}
    \dhsq(f,\hat{f}_n) &\lesssim \dhsq(f,f_\theta) + \dhsq(f_\theta,\hat{f}_n) \intertext{Consider two cases. If $\dhsq(f_\theta,\hat{f}_n) \ge 1/n$, then \Cref{thm:stronger-pensia} implies $\dhsq(f_\theta,\hat{f}_n) \lesssim \hdist^D(f_\theta,\hat{f}_n)$. Otherwise, we can use the upper bound of $1/n$. Hence:}
    & \lesssim \dhsq(f,f_\theta) + \hdist^D(f_\theta,\hat{f}_n) + \frac{1}{n} \\
    & \lesssim \dhsq(f,f_\theta) + \hdist^D(f,\hat{f}_n) + \frac{1}{n} \\
    & \lesssim \dhsq(f,f_\theta) + \hdist^D(f,f_n) + \hdist^D(f_n,\hat{f}_n)  + \frac{1}{n} \\
    & \lesssim \dhsq(f,f_\theta) + \hdist^D(f,f_n) + \hdist^D(f_n,f_\theta)  + \frac{1}{n} \rtag{definition of $\hat{f}_n$} \\
    & \lesssim \dhsq(f,f_\theta) + \hdist^D(f,f_n) + \hdist^D(f,f_\theta)  + \frac{1}{n} \\
    & \lesssim \dhsq(f,f_\theta) + \hdist^D(f,f_n) + \frac{1}{n} \intertext{We now invoke \Cref{cor:sqrt-converge} on each summand of $\hdist^D(f,f_n)$. A priori, we only have a bound on the VC dimension of the ratio class, which corresponds to where the ratio exceeds a threshold, yet each summand corresponds to where the ratio is within an interval. The relevant concept class is then the intersection of $\mathcal{H}$ with the complement of $\mathcal{H}$, which has VC dimension $\lesssim d$ (see e.g. Corollary 1.1 of \cite{van2009note}).}
    & \lesssim \dhsq(f,f_\theta) + D \cdot \frac{d \log(n) + \log(1/\delta)}{n} + \frac{1}{n} \\
    & \lesssim \dhsq(f,f_\theta) + \frac{(d \log(n) + \log(2/\delta)) \log(n)}{n}  \rtag{$D = \lceil \log(4n)\rceil$}
\end{align*}
\end{proof}

\section{Statistical corollaries from new recipe}\label{app:stat-corr}
We provide the proof of \Cref{cor:stat-log}; this matches up to logarithmic factors the known rate for learning log-concave densities (see \cite{doss2016global,KimSamworth}):
\begin{proof}
    We will choose $\mathcal{F}$ as the class of $t$-piecewise linear densities. The ratio class will be the union of $\lesssim t$ intervals, and hence has VC dimension $\lesssim t$. Using \Cref{thm:sharper-hel-min-dist}:
    \begin{align*}
        & \dhsq(f,\hat{f}_n) \lesssim \inf_{f_\theta \in \mathcal{F}} \dhsq(f,f_\theta)  + \frac{(t \log(n) + \ln \frac{2}{\delta}) \log(n)}{n} \intertext{Use the piecewise linear approximation guarantee of \Cref{lemma:piece-linear}, where each mixture component is approximated by a mixture of $t/k$ pieces. This implies any mixture of $k$ log-concave densities has a $t$-piecewise linear approximation within $\lesssim \frac{k^4}{t^4}$ squared Hellinger distance: }
        &  \lesssim \inf_{f_\theta \in \mathcal{LC}_k} \dhsq(f,f_\theta)  + \frac{k^4}{t^4}  + \frac{(t \log(n) + \ln \frac{2}{\delta}) \log(n)}{n} \\
        & \lesssim \inf_{f_\theta \in \mathcal{LC}_k} \dhsq(f,f_\theta) + \frac{k^{4/5} \log^{4/5}(n)}{n^{4/5}} + \frac{\log(2/\delta)\log(n)}{n}\rtag{choosing $t =\left( \frac{k^4 n}{\log(n)} \right)^{1/5}$}
    \end{align*}
\end{proof}

We provide the proof of \Cref{cor:stat-gaussian}; this matches up to logarithmic factors the known rate for learning mixtures of Gaussians:
\begin{proof}
    We will choose $\mathcal{F}$ as the class of $3k$-piecewise degree $\lesssim \log(n)$ polynomial densities. The ratio class will be the union of $\lesssim k\log(n)$ intervals, and hence has VC dimension $\lesssim k\log(n)$. Using \Cref{thm:sharper-hel-min-dist}:
    \begin{align*}
        & \dhsq(f,\hat{f}_n) \lesssim \inf_{f_\theta \in \mathcal{F}} \dhsq(f,f_\theta)  + \frac{(k\log^2(n) + \log \frac{2}{\delta}) \log(n)}{n} \intertext{We use existing results that imply any Gaussian density is approximated within $\eps$ squared Hellinger distance by a 3-piecewise degree $\lesssim \log(1/\eps)$ polynomial (Lemma 36 of \cite{chan2014efficient}, or related discussion in Section 5.4 of \cite{timan}). Accordingly, for any mixture of $k$ Gaussians, there is a $3k$-piecewise degree $\lesssim \log(n)$ polynomial that approximates the mixture within $1/n$ squared Hellinger distance: }
        &  \lesssim \inf_{f_\theta \in \mathcal{G}_k} \dhsq(f,f_\theta)  + \frac{1}{n} + \frac{(k\log^2(n) + \log \frac{2}{\delta}) \log(n)}{n} \\
        &  \lesssim \inf_{f_\theta \in \mathcal{G}_k} \dhsq(f,f_\theta) + \frac{(k\log^2(n) + \log \frac{2}{\delta}) \log(n)}{n}
    \end{align*}
\end{proof}

The proof of \Cref{cor:stat-mix} follows immediately from \Cref{thm:sharper-hel-min-dist} and Section 8.2 of \cite{devroye2001combinatorial} (since their proof already bounds the VC dimension of the corresponding ratio class); this matches up to logarithmic factors the known rate (e.g. an in-expectation bound is given by Corollary 7 of \cite{baraud2011estimator}, and high-probability guarantees are given by $\rho$-estimation or sharp bounds up to constant factors in \cite{compton2026ratio}).

\section{Details for learning log-concave mixtures in near-linear time}\label{app:log-concave}
The proofs of our result \Cref{corr:fast-gauss} for mixtures of Gaussians, and \Cref{thm:fast-log-concave} for mixtures of log-concave densities, are slightly different because we will save some logarithmic factor for the log-concave result (via employing a sharper reverse data processing inequality). In this section, we will give the details for our log-concave result, and later in \Cref{app:piecewise-poly} we will give details for more general piecewise polynomial results (which will use many results given in this section).

First, in \Cref{subsec:bounding-hel-from-h1}, we will prove that approximately minimizing $\mathcal{H}_1$-distance (a distance corresponding to the concept class of intervals) implies an upper bound on the Hellinger distance. This will roughly follow the structure of the proof of \Cref{thm:hel-min-dist}, except non-trivial modifications will be made so that we may minimize $\mathcal{H}_1$-distance instead of some more complicated $\hdist$ metric. Interestingly, in this specific setting we will be able to leverage a tailored version of the reverse data processing inequality that will not lose a logarithmic factor.

Second, in \Cref{subsec:greedy-merge}, we will outline the greedy merging method of \cite{acharya2017sample}. Suppose you are given an oracle that computes a degree-$d$ polynomial which approximately minimizes the $\mathcal{H}_1$ distance with the empirical distribution restricted to an interval $\mathcal{I}$. Then, with this oracle, the greedy merging method of \cite{acharya2017sample} will yield a method for constructing a $2t$-piecewise degree-$d$ polynomial with $\mathcal{H}_1$ distance from the empirical distribution that is not much larger than the best $t$-piecewise degree-$d$ polynomial. After a minor modification, this will yield a near-linear time algorithm that approximately minimizes $\mathcal{H}_1$-distance. 

Third, in \Cref{subsec:1-piece-opt}, we will construct the required oracle. The procedure will follow from the machinery of  \cite{acharya2017sample} that efficiently determines an approximately optimal polynomial by leveraging fast separation oracles. Our key difficulty is to provide a new separation oracle for $\mathcal{H}_1$ distance: given a certain polynomial, we must find a way to approximately compute the $\mathcal{H}_1$ distance from the empirical samples in near-linear time.

Both these second and third parts are presented in terms of degree-$d$ polynomials because this will be helpful for further applications in \Cref{app:piecewise-poly}; although in this section we only need $d=1$. Finally, in \Cref{subsec:combine-ingreds} we combine our ingredients to conclude \Cref{thm:fast-log-concave}.

\subsection{Bounding squared Hellinger distance from $\mathcal{H}_1$ distance}\label{subsec:bounding-hel-from-h1}
Our goal will be to output a piecewise-linear estimate $\hat{f}_n$ which has small $\mathcal{H}_1$ distance from the empirical samples. We will show that such an estimate has small squared Hellinger distance with the true distribution $f$. First, let us formally define $\mathcal{H}_1$ distance:

\begin{equation*}
    \mathcal{H}_1(f,g) \triangleq \sup_{l \le r} \frac{1}{2} \left( \sqrt{\Pr_f\left[x \in [l,r] \right]} - \sqrt{\Pr_g\left[x \in [l,r] \right]} \right)^2
\end{equation*}

We will find it helpful to more generally define $\mathcal{H}_i$ distance:

\begin{equation*}
    \mathcal{H}_i(f,g) \triangleq \sup_{l_1 \le r_1 < l_2 \le r_2 < \dots < l_i \le r_i} \frac{1}{2} \sum_{j=1}^i\left( \sqrt{\Pr_f\left[x \in [l_j,r_j] \right]} - \sqrt{\Pr_g\left[x \in [l_j,r_j] \right]} \right)^2
\end{equation*}

Following the outline of \Cref{thm:hel-min-dist}, we will need results showing that log-concave mixtures are well-approximated by piecewise-linear functions (see \Cref{sec:piecewise-linear}), and a custom reverse data processing inequality that shows $\dhsq(f,g) = \Theta(1) \cdot  \mathcal{H}_k(f,g)$ for the piecewise-linear densities of interest (see \Cref{sec:tight-rdpi}). This result will leverage that log-concave densities may be well-approximated by a restricted class of piecewise-linear densities. We define this restricted class as $\lin_t$, which consists of all $t$-piecewise linear densities where the value of each line at the start of its piece is within a factor of two of its value at the end of the piece.

Given these two tools, we are ready to state the result:
\begin{lemma}\label{lemma:h1-to-err}
    Let $f$ be distribution from which we receive $n$ samples. Suppose $\hat{f}_n \in \lin_{rt}$ (for some positive constant integer $r$) is a density where $\mathcal{H}_1(\hat{f}_n,f_n) \le \alpha \cdot \inf_{f_\theta \in \lin_{2t}} \mathcal{H}_1(f_\theta,f_n) + \beta$. Then with probability $1-\delta$,
    \begin{equation*}
        \dhsq(f,\hat{f}_n) \lesssim  \alpha \cdot \inf_{f_\theta \in \lin_t} \dhsq(f,f_\theta) +  \alpha t \cdot \frac{\log(n) + \log(2/\delta)}{n} +  \beta
    \end{equation*}
\end{lemma}
\begin{proof}
    Consider any $f_\theta \in \lin_t$. It holds that
    \begin{align*}
        \dhsq(f,\hat{f}_n) &\lesssim  \dhsq(f,f_\theta) + \dhsq(f_\theta,\hat{f}_n) \intertext{Observe how for any two densities in $\lin_{rt}$, the domain can be decomposed into $2rt$ intervals satisfying the conditions of our custom reverse data processing result of \Cref{lemma:custom-rdpi}. This means that for each of the $\le 2rt$ pieces where $f_\theta$ and $\hat{f}_n$ are each one line, there is an interval that captures a constant-fraction of the squared Hellinger distance over that piece. This yields:}
        & \lesssim  \dhsq(f,f_\theta) + \mathcal{H}_{2rt}(f_\theta,\hat{f}_n) \\
        & \lesssim  \dhsq(f,f_\theta) + \mathcal{H}_t(f_\theta,\hat{f}_n) \rtag{via $\mathcal{H}_{2rt}(a,b) \le 2r \mathcal{H}_t(a,b)$} \\
        & \lesssim  \dhsq(f,f_\theta) + \mathcal{H}_t(f,\hat{f}_n) \rtag{via $\mathcal{H}_t(a,b) \le \dhsq(a,b)$} \\
        & \lesssim \dhsq(f,f_\theta) + \mathcal{H}_t(f,f_n)  + \mathcal{H}_t(f_n,\hat{f}_n) \\
        & \lesssim \dhsq(f,f_\theta) + t \mathcal{H}_1(f,f_n)  + t \mathcal{H}_1(f_n,\hat{f}_n) \rtag{via $\mathcal{H}_t(a,b) \le t \mathcal{H}_1(a,b)$} \numberthis \label{eq:intermediate-lin}
    \end{align*}
    At this point, we would like to design a modification of $f_\theta$ where each piece of the modified version has the same integral as $f$ over its domain, and the Hellinger distance contribution of each piece is roughly evenly-split. Later, this modification will be helpful because it will have small $\mathcal{H}_1$ distance from $f$.

    \begin{lemma}\label{lemma:tilde-f}
        Consider any two non-negative functions $p,q$, where $q$ is a $k$-piecewise-$\mathcal{F}$ function, and $\mathcal{F}$  is a family closed under rescaling. Then, there is a $(2k-1)$-piecewise-$\mathcal{F}$ function $\tilde{q}$, where $\dhsq(p,\tilde{q}) \lesssim \dhsq(p,q)$,  $\max_{\textrm{piece } \mathcal{I} \textrm{ of $\tilde{q}$}} \dhsq(p_\mathcal{I},\tilde{q}_\mathcal{I}) \lesssim \frac{\dhsq(p,q)}{k}$, and the integral of $p$ and $\tilde{q}$ match over every piece of $\tilde{q}$.
    \end{lemma}
    \begin{proof}
        We will construct $\tilde{q}$ by modifying $q$. Consider adding $k-1$ breakpoints at the $k$-quantiles of $\dhsq(p,q)$: meaning the squared Hellinger distance contribution from each piece is now at most $\frac{1}{k}$ fraction of the original Hellinger distance $\dhsq(p,q)$. Then, for each of these pieces in $\tilde{q}$, we will rescale the piece such that its integral is exactly the same as the integral of $p$ over the corresponding interval. This is our $\tilde{q}$: all that remains is to show $\dhsq(p,\tilde{q})$ is sufficiently small. We will do so by proving the squared Hellinger distance between $p$ and $\tilde{q}$ over any piece of $\tilde{q}$, is not much larger than the squared Hellinger distance between $p$ and $q$ in that interval.

        \begin{claim}\label{claim:rescale-ok}
            Given any two non-negative functions $f,g$ with total mass $w_f,w_g > 0$, then $\dhsq(f,\frac{w_f}{w_g} g) \lesssim \dhsq(f,g)$.
        \end{claim}
        \begin{proof}
            Let us separate into two cases: (i) the ratio $\max(\frac{w_f}{w_g},\frac{w_g}{w_f}) \ge 1.5$, and (ii) otherwise.
            
            \underline{Case (i): the ratio $\max(\frac{w_f}{w_g},\frac{w_g}{w_f}) \ge 1.5$.} First, we provide an upper bound for $\dhsq(f,\frac{w_f}{w_g}g)$:
            \begin{equation*}
                \dhsq(f, \frac{w_f}{w_g} g) \lesssim \tv (f, \frac{w_f}{w_g} g) \lesssim w_f
            \end{equation*}
            Second, we provide a lower bound for $\dhsq(f,g)$. Without loss of generality, let $A$ be the function with larger total mass, and $B$ be the other function. Define $S$ as the subset of the domain where $A(x) \ge 1.1 B(x)$. It must be the case that $\int_S A(x) \ge 0.1 \cdot w_A$. Otherwise $w_B \ge \frac{0.9 w_A}{1.1}$ which violates our ratio casework. Using this: 
            \begin{equation*}
                \dhsq(f,g) = \dhsq(A,B) \gtrsim  \int_S \left(\sqrt{A(x)} - \sqrt{B(x)} \right)^2 \gtrsim \int_S A(x) \gtrsim w_A = \max(w_f,w_g)
            \end{equation*}

            Combining our upper and lower bounds, we immediately conclude $\dhsq(f, \frac{w_f}{w_g} g) \lesssim \dhsq(f,g)$.

            \underline{Case (ii): the ratio $\max(\frac{w_f}{w_g},\frac{w_g}{w_f}) < 1.5$.} In this case, let us say $\frac{w_f}{w_g} = 1+r$ for some $|r| \le \frac{1}{2}$.
            \begin{align*}
                \dhsq(f,\frac{w_f}{w_g} g) &= \frac{1}{2} \int \left( \sqrt{f(x)} - \sqrt{(1+r) g(x)} \right)^2 \\
                &\lesssim \int \left( \sqrt{f(x)} - \sqrt{g(x)} \right)^2  + \int \left( \sqrt{(1+r) g(x)} - \sqrt{ g(x)} \right)^2 \\ 
                & \lesssim \dhsq(f,g) + \left(\sqrt{(1+r) \cdot w_g} - \sqrt{w_g} \right)^2 \\
                & = \dhsq(f,g) + \left(\sqrt{w_f} - \sqrt{w_g} \right)^2 \\
                & \lesssim \dhsq(f,g) \rtag{via \Cref{fact:dpi}} %
            \end{align*}
        \end{proof}

        We may use this to complete the proof of our lemma. For any interval $\mathcal{I}$ corresponding to a piece of $\tilde{q}$, we know $\dhsq(p_\mathcal{I},q_\mathcal{I}) \le \frac{\dhsq(p,q)}{k}$. Combining with \Cref{claim:rescale-ok}, we know 
        \begin{equation*}
            \dhsq(p_\mathcal{I},\tilde{q}_\mathcal{I}) \lesssim \dhsq(p_\mathcal{I},q_\mathcal{I}) \le \frac{\dhsq(p,q)}{k}%
        \end{equation*}
    \end{proof}

    Let's resume from \Cref{eq:intermediate-lin} with this modified density $\tilde{f_\theta}$ in hand. By definition, we have a bound for the distance $\mathcal{H}_1(f_n,\hat{f}_n)$ in terms of any other density in $\lin_{2t}$, including $\tilde{f_\theta}$:
    \begin{align*}
        & \dhsq(f,f_\theta) + t \mathcal{H}_1(f,f_n)  + t \mathcal{H}_1(f_n,\hat{f}_n)  \lesssim \dhsq(f,f_\theta) + t \mathcal{H}_1(f,f_n)  + \alpha t \mathcal{H}_1(f_n,\tilde{f_\theta}) + \beta\\
        & \lesssim   \dhsq(f,f_\theta) +  \alpha t \mathcal{H}_1(f,f_n)  +  \alpha t \mathcal{H}_1(f,\tilde{f_\theta})+  \beta \intertext{Leveraging \Cref{fact:dpi} and how the integral of $f$  and $\tilde{f_\theta}$ is the same for any piece of $\tilde{f_\theta}$, we observe the distance $\mathcal{H}_1(f,\tilde{f_\theta})$ is at most a constant factor larger than the same $\mathcal{H}_1$ distance when we restrict to intervals within one piece of $\tilde{f_\theta}$. This quantity is at most the maximum squared Hellinger distance between $f$ and $\tilde{f_\theta}$ over one piece. By the properties of $\tilde{f_\theta}$ given by \Cref{lemma:tilde-f}, this implies a bound of $\mathcal{H}_1(f,\tilde{f_\theta}) \lesssim \frac{1}{t} \cdot \dhsq(f,f_\theta)$:}
        & \lesssim   \alpha \dhsq(f,f_\theta) +  \alpha t \mathcal{H}_1(f,f_n)+  \beta \\
        & \lesssim   \alpha \dhsq(f,f_\theta) +  \alpha t \cdot \frac{\log(n) + \log(2/\delta)}{n} +  \beta
    \end{align*}
    The last line holds with probability $1-\delta$ from the uniform convergence guarantee \Cref{cor:sqrt-converge}.
\end{proof}

\subsubsection{Piecewise-linear approximation}\label{sec:piecewise-linear}

We now present the required piecewise-linear approximation result for log-concave densities. This result follows from the proof of Theorem 12 of \cite{diakonikolas2016efficient}:
\begin{lemma}[Piecewise-linear approximation for log-concave; follows from \cite{diakonikolas2016efficient}]\label{lemma:piece-linear} Let $f  : I \rightarrow \mathbb{R}_+$ be any log-concave density, where $I = [a,b]$ is an arbitrary (not necessarily finite) interval. For any $\eps \in (0,\nicefrac{1}{2})$, there exists a $t$-piecewise linear density $\tilde{f}$ where $t = O(1/{\eps^{1/4}})$ and $\dhsq(f,\tilde{f}) \le \eps$. Moreover, for each piece of $\tilde{f}$, the values of $\tilde{f}(x)$ are within a factor of $2$ at the start/end of the piece.
\end{lemma}
\begin{proof}
    The proof will follow almost exactly the same analysis given in Theorem 12 of \cite{diakonikolas2016efficient}. We will assume familiarity with this proof, and now discuss the small modifications we require. We will aim to modify their analysis to get the guarantee of $O(\eps^2)$ squared Hellinger distance with $t = O(1/{\sqrt{\eps}})$ pieces. This would immediately imply our desired guarantee.

    First, we will modify the guarantee of their Lemma 14 such that $\dhsq(f,g) \lesssim \eps^2 \| f \|_1$. After modifying a small typo in the final equation of their proof of Lemma 14, we may conclude $|f(x) - g_i(x) | \lesssim \eps \cdot f(x)$. We provide the modified equation:
    \begin{align*}
        &f(x) = (1 + O(\eps) ) f(y) (1 + (x-y) \alpha_0 + O((x-y) \alpha_0)^2) \\
        &= (1 + O(\eps) ) f(y) (1 + (x-y) \alpha_0 + O(\eps)) = f(y) + f(y) \cdot (x-y) \cdot \alpha_0 + f(y) \cdot O(\eps)
    \end{align*}

    Hence, with $g_i(x) = f(y) + f(y) \cdot (x-y) \alpha_0$, we surmise $|g_i(x) - f(x)| \lesssim f(x) \cdot \eps$. This immediately implies $\dhsq(f,g) \lesssim \eps^2 \|f\|_1$.

    Their proof of Proposition 15 remains unchanged. 

    Finally, how they invoke the modified Lemma 14 and Proposition 15 is almost exactly the same as what we require. \textit{We now repeat exactly their proof, only modifying parameters and replacing $L_1$ error bounds with the analogous squared Hellinger error bounds.} The remaining proof replaces their proof section starting with ``for $m=1,\dots,2\log(1/\eps)/c$, we use Lemma 14 ...''.
    
    For $m=1,\ldots,2\log(1/\eps)/c$, we use Lemma 12 to approximate $f$ in $I_m$ by a piecewise linear function $g_m$ so that $g_m$ has at most $O(\eps^{-1/2}2^{-cm/8})$ pieces and so that the squared Hellinger distance between $f$ and $g_m$ on $I_m$ is at most $f(I_m)O(\eps^2 2^{cm/2})= O(\eps^2 2^{-cm/2}).$ Let $g$ be the piecewise linear function that is $g_m$ on $I_m$ for $m\leq c\log(1/\eps)/2$, and $0$ elsewhere. $g$ is piecewise linear on
    $$
    \sum_{m=1}^{2\log(1/\eps)/c} O(\eps^{-1/2}2^{-cm/8}) = O(\eps^{-1/2})
    $$
    intervals.
    
    Furthermore the squared Hellinger error between $f$ and $g$ on the $I_m$ with $m\leq 2\log(1/\eps)/c$ is at most
    $$
    \sum_{m=1}^{2\log(1/\eps)/c} O(\eps^2 2^{-cm/2}) = O(\eps^2 ).
    $$
    
    The squared Hellinger error from other intervals is at most
    $$
    \sum_{m=2\log(1/\eps)/c}^\infty O(f(I_m)) = \sum_{m=2\log(1/\eps)/c}^\infty O(2^{-cm}) = O(\eps^2).
    $$
    Therefore, $\dhsq(f,g)=O(\eps^2)$. 
    
    We finally remark that each piece of the density satisfies our bounded ratio condition between the start/end by their construction.
\end{proof}

\subsubsection{Tighter reverse data processing inequality}\label{sec:tight-rdpi}

\begin{lemma}\label{lemma:custom-rdpi}
    Consider any two densities $p,q$ that are linear densities supported on (without loss of generality) $[0,1]$. Further assume a ratio condition between the endpoints: $\max(\frac{p(0)}{p(1)},\frac{p(1)}{p(0)},\frac{q(0)}{q(1)},\frac{q(1)}{q(0)}) \le 2$. Then, there exists some interval $\mathcal{I} \subseteq [0,1]$ where
    \begin{equation*}
        \left( \sqrt{\int_\mathcal{I} p(x) } - \sqrt{\int_\mathcal{I} q(x)} \right)^2 \gtrsim \dhsq(p,q)
    \end{equation*}
\end{lemma}

Note how here we are proving a reverse data processing inequality that crucially only loses a constant fraction of squared Hellinger distance, instead of a $\log(2/\dhsq(p,q))$ term. 

\begin{proof}

Without loss of generality, suppose $p(x) \ge q(x)$ in the support. (If $p,q$ intersect, we can break the density into two intervals, and focus on the interval with at least half of the original squared Hellinger distance.) Let $r(x) = \frac{p(x)}{q(x)}$. From earlier, we know $r(x) \ge 1$.

If $p(x) = ax + b$ and $q(x) = cx+d$, then observe how $r'(x) = \frac{ad-bc}{q(x)^2}$ is monotonic. Without loss of generality, we flip the density such that $r(x)$ is non-decreasing. From our ratio condition between the endpoints, we know:

\begin{equation*}
    \frac{\max_{x \in [0,1]} r'(x)}{\min_{x \in [0,1]} r'(x)} =\frac{\max_{x \in [0,1]} q(x)^2}{\min_{x \in [0,1]} q(x)^2} \le 4
\end{equation*}

Then, for ease of notation we will denote $v \triangleq \max_{x \in [0,1]} r'(x)$ and surmise that for all $x \in [0,1]$ we know $v/4 \le r'(x) \le v$.

Similar to the proof of the original reverse data processing inequality (Appendix~\ref{sec:pensia}), we will consider two cases: (i) when most distance comes from the summand where the ratio is less than two, or (ii) where most distance comes from the summand with ratio at least two. Case (i) will leverage structure from the linear densities, while case (ii) will follow exactly as in the proof of the original reverse data processing inequality.

With this in mind, let us define $x^*$ as the smallest value in  $[0,1]$ where $r(x) \ge 2$ (if $r(x) < 2$ for all $x \in [0,1]$, then set $x^*=1$). By definition, we know
\begin{equation}
    \dhsq(p,q) = \frac{1}{2}\int_0^{x^*} \left( \sqrt{p(x)} - \sqrt{q(x)}  \right)^2 +  \frac{1}{2} \int_{x^*}^{1} \left( \sqrt{p(x)} - \sqrt{q(x)}  \right)^2
\end{equation}

This will let us focus on our two cases.

\textbf{Case (i): $\int_0^{x^*} \left( \sqrt{p(x)} - \sqrt{q(x)}  \right)^2 \ge \dhsq(p,q)$.} In this case, our plan is to choose the interval $\mathcal{I} = [x^*/2,x^*]$ and then bound the desired quantities. First, let us upper bound the squared Hellinger distance term:

\begin{align*}
    \int_0^{x^*} \left( \sqrt{p(x)} - \sqrt{q(x)}  \right)^2 &=  \int_0^{x^*} \left( \sqrt{r(x) \cdot q(x)} - \sqrt{q(x)}  \right)^2  \\
    & \le \int_0^{x^*} \left( \sqrt{(r(0) + v x^*) \cdot q(x)} - \sqrt{q(x)}  \right)^2   \\
    & \le (r(0) + vx^* - 1)^2  \cdot \int_0^{x^*} q(x) \label{eq:custom-rdpi-ub}\numberthis
\end{align*}

Next, let us lower bound the distance yielded by the chosen interval:
\begin{align*}
    &\left( \sqrt{\int_{x^*/2}^{x^*} p(x) } - \sqrt{\int_{x^*/2}^{x^*} q(x)} \right)^2 \\
    &= \left( \sqrt{\int_{x^*/2}^{x^*} r(x) \cdot q(x) } - \sqrt{\int_{x^*/2}^{x^*} q(x)} \right)^2 \\
    & \ge \left( \sqrt{\int_{x^*/2}^{x^*} (r(0) + vx^*/8) \cdot q(x) } - \sqrt{\int_{x^*/2}^{x^*} q(x)} \right)^2\\
    & \ge (r(0) + vx^*/8 - 1)^2/9 \cdot \int_{x^*/2}^{x^*} q(x) \rtag{using $\sqrt{1+z} \ge 1+z/3$ for $z \in [0,1]$} \label{eq:custom-rdpi-lb} \numberthis
\end{align*}

Hence, combining \Cref{eq:custom-rdpi-ub,eq:custom-rdpi-lb} we conclude:
\begin{align*}
    &\left( \sqrt{\int_{x^*/2}^{x^*} p(x) } - \sqrt{\int_{x^*/2}^{x^*} q(x)} \right)^2 \\
    &\gtrsim \dhsq(p,q) \cdot \frac{\left( \sqrt{\int_{x^*/2}^{x^*} p(x) } - \sqrt{\int_{x^*/2}^{x^*} q(x)} \right)^2}{\int_0^{x^*} \left( \sqrt{p(x)} - \sqrt{q(x)}  \right)^2} \\
    & \ge \dhsq(p,q) \cdot \frac{(r(0) + vx^*/8 - 1)^2/9 \cdot \int_{x^*/2}^{x^*} q(x)}{(r(0) + vx^* - 1)^2  \cdot \int_0^{x^*} q(x)} \\
    & \gtrsim  \dhsq(p,q) \cdot \frac{\int_{x^*/2}^{x^*} q(x)}{\int_{0}^{x^*} q(x)}  \rtag{using for non-negative $a,b,c,d>0$ that $\frac{a+b}{a+c} \le 1 + \frac{b}{c} \implies \frac{a+b}{a+c} \ge \frac{1}{1 + \frac{c}{b}}$}\\
    & \gtrsim \dhsq(p,q) \rtag{via ratio condition between endpoints for $q$}
\end{align*}

\textbf{Case (ii): $ \int_{x^*}^{1} \left( \sqrt{p(x)} - \sqrt{q(x)}  \right)^2 \ge \dhsq(p,q)$.} We will choose $\mathcal{I} = [x^*,1]$ and our claim will follow by the same method as the original proof of the reverse data processing inequality:
\begin{align*}
    & \left( \sqrt{\int_{x^*}^{1} p(x) } - \sqrt{\int_{x^*}^{1} q(x)} \right)^2  = \left( \sqrt{\int_{x^*}^{1} p(x) } - \sqrt{\int_{x^*}^{1} p(x) / r(x)} \right)^2 \ge \left( \sqrt{\int_{x^*}^{1} p(x) } - \sqrt{\int_{x^*}^{1} p(x) / 2} \right)^2 \\
    & \gtrsim \int_{x^*}^{1} p(x) \ge \int_{x^*}^{1} \left( \sqrt{p(x)} - \sqrt{q(x)} \right)^2 \ge \dhsq(p,q) %
\end{align*}
\end{proof}

\textit{Remark on conditions.} We expect the ratio condition between the endpoints is crucial for this modified reverse data processing inequality where we only lose a constant fraction of squared Hellinger distance, instead of a $\log(2/\dhsq(p,q))$ term. Informally, we suggest considering the densities $p(x) = 2+x$ and $q(x) = 1+x$ (distributed over $[0,z]$ for a very large value of $z$, and then normalized) for a counter-example without this ratio condition.

\subsection{Greedy merging}\label{subsec:greedy-merge}
This section will outline the greedy merging procedure in Section 4 of \cite{acharya2017sample}. We will state some of their results slightly more generally than they state in their work, but the results hold from the same logic given in their proof.

The main purpose of the greedy merging procedure is to quickly optimize over $t$-piecewise functions when we only have an oracle that is able to optimize over single-piece functions. We now introduce the relevant definitions/oracles and then describe their greedy merging procedure.

\begin{definition}[Non-decreasing cost]
    Consider a cost function $L$ over intervals defined with respect to a family of densities $\mathcal{F}$. Suppose the function $L(\mathcal{I},f_n,\hat{f})$ assigns a non-negative cost for an interval $\mathcal{I}$ depending on the empirical samples $f_n$ and an $\hat{f} \in \mathcal{F}$. Further, we define $L(\mathcal{I},f_n) = \inf_{\hat{f} \in \mathcal{F}} L(\mathcal{I},f_n,\hat{f})$. Then, we say $L$ is \textit{non-decreasing} if $L(\mathcal{I}_1) \le L(\mathcal{I}_2)$ whenever $\mathcal{I}_1 \subseteq \mathcal{I}_2$.
\end{definition}

For intuition, you should think of $\mathcal{F}$ as a class of single-piece functions (e.g. degree-$d$ polynomials). We define an oracle that approximately computes $L$ and computes an approximate optimizer $\hat{f} \in \mathcal{F}$:

\begin{definition}[$\mathcal{O}_{(\alpha,\beta)-opt}(\mathcal{I},f_n)$ oracle]\label{def:oracle-find}
    We define an oracle $\mathcal{O}_{(\alpha,\beta)-opt}(\mathcal{I},f_n)$ as one that takes in empirical samples and an interval $\mathcal{I}$, and then outputs: (i) $\hat{\ell}$, an estimate of the cost, and (ii) $\hat{f} \in \mathcal{F}$, an approximate optimizer. In particular, it should hold that $\hat{\ell} \le L(\mathcal{I},f_n) \le L(\mathcal{I},f_n,\hat{f}) \le \alpha \hat{\ell} + \beta$. Further, we define $T_{(\alpha,\beta)-opt}(n)$ as an upper bound on the maximum total time it would take to run $\mathcal{O}_{(\alpha,\beta)-opt}(\mathcal{I},f_n)$ on at most $n$ samples partitioned into at most $2n$ intervals.
\end{definition}

With these definitions, we may state the guarantee of the greedy merging procedure of \cite{acharya2017sample}:

\begin{theorem}[Guarantees for greedy merging of \cite{acharya2017sample} Section 4]\label{thm:merge-guarantee}
Consider an interval $[a,b]$ and $n$ samples where $a \le x_1 \le \dots \le x_n \le b$. Let $\operatorname{OPT}_t$ denote minimum-possible maximum cost of an interval if we partition $[a,b]$ into $t$ pieces. Suppose for an initial partition $[a,x_1],(x_1,x_2], \dots, (x_n,b]$, it holds that the cost of each interval is at most $\beta$. Then, \Cref{alg:greedy-merge} runs in $O(T_{(\alpha,\beta)-opt}(n) \log(n))$ time and produces a partition into at most $2t$ intervals where the maximum cost of an interval is at most $\alpha \operatorname{OPT}_t + \beta$.    
\end{theorem}

\begin{algorithm}[htb]
\caption{Greedy merging algorithm (Algorithm 2 in \cite{acharya2017sample}).}
\label{alg:greedy-merge}
\SetAlgoNoEnd
\DontPrintSemicolon
\LinesNumbered
\SetKwProg{Fn}{Function}{:}{}
\SetKwInOut{KwIn}{Input}

\KwIn{Interval $[a,b]$, piece parameter $t$, and $n$ samples where
$a \le x_1 \le \dots \le x_n \le b$.}

\Fn{\textsc{Greedy-Merging}$([a,b], t, (x_1,\dots,x_n))$}{

Form the empirical distribution $f_n$ from these samples.\;

Let $\mathcal{I}_0 \gets \{ [a,x_1], (x_1,x_2], (x_2,x_3], \dots, (x_n,b] \}$ be the initial partition. $j \gets 0$\;
  
  \While{$|\mathcal{I}_j| > 2t$}{
    Let $\mathcal{I}_j = \{I_{1,j}, I_{2,j}, \ldots, I_{s_j-1,j}, I_{s_j,j}\}$.\;

    \For{$k \gets 1$ \KwTo $\left\lceil \frac{s_j}{2} \right\rceil$}{
      
      $I'_{k,j+1} \gets I_{2k-1,j} \cup I_{2k,j}$\;
      \tcp{Set this to $I_{s_j,j}$ for last $k$ if $s_j$ is odd}

      $\hat{\ell}_{k,j}, \hat{f}_{k,j} \gets \mathcal{O}_{(\alpha,\beta)\text{-opt}}(I'_{k,j+1}, f_n)$\;
    }

    Let $S$ be the set of $k \in \{1,2,\ldots,\lceil s_j/2\rceil\}$ with the $t-1$ largest costs $\hat{\ell}_{k,j}$.\;
    
    Let $M$ be the complement of $S$.\;

    $\mathcal{I}_{j+1} \gets \bigcup\limits_{k \in S} \{I_{2k-1,j}, I_{2k,j}\}$\;
    \tcp{Do not merge the intervals in $S$}

    $\mathcal{I}_{j+1} \gets \mathcal{I}_{j+1} \cup \{I'_{k,j+1} \mid k \in M\}$\;
    \tcp{Merge the intervals in $M$}

    $j \gets j+1$\;
  }

  \Return{$\mathcal{I}=\mathcal{I}_j$ and the functions $\hat{f}$ for all intervals in $\mathcal{I}$.}\;
}
\end{algorithm}

\textbf{Algorithm description and proof sketch. } In each phase, \Cref{alg:greedy-merge} considers merging pairs of consecutive intervals. Before doing so, we use the oracle to approximately compute the cost of the potentially merged intervals. By merging only proposed intervals whose computed costs were less than the $t-1$ largest costs, it is our goal to upper bound the cost on new intervals in relation to $OPT_t$.

Consider some optimal $t$-piece interval decomposition corresponding to the solution value of $\operatorname{OPT}_t$; naturally, this decomposition has at most $t-1$ breakpoints between pieces. This means that one of the $t$ largest cost intervals computed by the oracle must not contain a breakpoint. Call one such interval $I^*$; we know $L(I^*,f_n)\le OPT_t$ by the definition of non-decreasing cost, since $I^*$ is a subset of an interval in an optimal decomposition. This is sufficient for us to upper bound the loss of any new interval $J$ that is ever merged in $M$. For ease of notation, for any interval $I$ let us denote the $\hat{f}$ computed by the oracle as $\hat{f}(I)$, and similarly the approximate cost as $\hat{\ell}(I)$. Then, for any newly merged interval $J \in M$:

\begin{equation*}
    L(J,f_n,\hat{f}(J)) \le \alpha \hat{\ell}(J) + \beta \le \alpha \hat{\ell}(I^*) + \beta \le \alpha L(I^*,f_n)+\beta \le \alpha OPT_t + \beta 
\end{equation*}

Thus, since every time we merge intervals it forms an interval with $\hat{f}$ costing at most $\alpha \opt_t + \beta$, and since we know all the original intervals start with cost at most $\beta$, then we may conclude the maximum cost $\hat{f}$ for an interval in our final output is at most $\alpha \opt_t + \beta$.

The runtime follows because the number of phases is bounded by $O(\log(n))$. This is because in each phase, all but at most $2t-1$ intervals are merged (halving the number of excess intervals).

\subsection{$\mathcal{H}_1$ optimization over 1-piece functions}\label{subsec:1-piece-opt}

We now turn towards constructing the oracle given in \Cref{def:oracle-find}, where we hope to compute a 1-piece function with approximately minimal $\mathcal{H}_1$ distance from the empirical distribution restricted to a certain interval. Our algorithm will follow directly from the algorithm and analysis given in Sections 5 and 6 of \cite{acharya2017sample}; the only key difference is that we will require a different separation oracle. In this section, we will state the theorem implied by the algorithm/analysis of \cite{acharya2017sample}, we will then describe their approach, and then we will design the modified separation oracle we require:

\begin{theorem}[analogous to Theorem 18 of \cite{acharya2017sample}]\label{thm:p-opt}

    Fix an interval $J \subseteq [-1,1]$. Consider $\mathcal{F}_d$: the class of degree-$d$ densities over $J$ (in the case of $d=1$, we may further restrict to lines where the value at the endpoints of $J$ are within a factor of two of each other). There exists an $\mathcal{H}_1$ oracle $\mathcal{O}_{(O(1),O(1/n))-opt}$ for $\mathcal{F}_d$ that runs in time
    \begin{equation*}
        O \left( (d^3 \log \log n + sd^2 + d^{\omega + 2}) \log^2 (n) \right)
    \end{equation*}
    where $s$ is the number of samples in the interval $J$.
\end{theorem}

\textbf{Description of the approach of \cite{acharya2017sample}.} We redescribe in more detail the approach in \cite{acharya2017sample}. In their analogous result, they are searching for a polynomial whose density has small $\mathcal{A}_k$ distance with the empirical distribution over $J$; for this discussion, let us just describe the case of $\mathcal{A}_1$. We will do so by approximately determining whether there is a polynomial with $\mathcal{A}_1$ distance at most $\tau$, and employing a binary search over the correct value of $\tau$.

Any degree-$d$ polynomial can be defined by its coefficients in $\mathbb{R}^{d+1}$. Let $\mathcal{C}_\tau \subset \mathbb{R}^{d+1}$ denote the set of such polynomials. These polynomials must: (i) be non-negative over $J$, and (ii) must satisfy an upper and lower bound on the integral of $p$ for every $I \subset J$ (meaning, based on $\tau$ and the number of samples in $I$, there is a minimum/maximum possible integral for $p$ over $I$). Observe how each of these two types of constraints are linear in the coefficients of $p$. The key technical difficulty is that the number of these constraints may be quite large; there are many values of $x$ in the domain for constraints of type (i), and naively there are $\Theta(s^2)$ relevant intervals producing constraints of type (ii), yet we hope for a running time that is near-linear in terms of $s$.

Accordingly, their work does not explicitly define all these constraints. Instead, for some fixed polynomial $p$, we either find a constraint that is violated (thus giving a halfspace that separates $p$ from $\mathcal{C}_\tau$, the primitive enabling e.g. the cutting plane method of \cite{vaidya1989new}), or we conclude $p$ at least approximately satisfies all constraints. Given efficient subroutines that can quickly find such a violated constraint, then \cite{acharya2017sample} are able to conclude the desired algorithm to find a near-optimal $p$.

We will not restate the details of \cite{acharya2017sample}, as we only need to modify a subroutine. First, in the special-case where $d=1$ and we want the endpoints to satisfy the bounded ratio constraint for $\lin$, we note this can be encoded as two linear constraints that we explicitly check at each stage of the algorithm. We do not modify their machinery for handling the polynomial non-negativity constraints of type (i). We only modify the procedure for constraints of type (ii), where instead we must now check whether $p$ has bounded $\mathcal{H}_1$ distance from the empirical distribution (instead of $\mathcal{A}_1$). The subroutine for $\mathcal{A}_1$ distance mostly followed from a previously studied related question of \cite{csuros2004maximum}; for $\mathcal{H}_1$ distance we will need an alternative algorithm.

\subsubsection{$\mathcal{H}_1$ separation oracle}

In this section, we will prove the following result, which when used as a separation oracle for the algorithm of \cite{acharya2017sample} will imply the guarantees of \Cref{thm:p-opt}. Note that the analogous result in \cite{acharya2017sample} for $\mathcal{A}_1$ distance only has additive error, and this is why our $\mathcal{H}_1$ oracle also incurs multiplicative error. Additionally, since the runtime of this subroutine is $O(s)$ it is not a dominating term in the final runtime.

\begin{theorem}
Fix an interval $J =[a,b] \subseteq [-1,1]$. Suppose there are $s$ samples in this range $a \le x_1 \le \dots \le x_s \le b$. For some non-negative polynomial $p$, suppose we are given access to $c_0,\dots,c_s$ where $c_0 = \int_a^{x_1} p(t) \, dt$,  $c_1 = \int_{x_1}^{x_2} p(t) \, dt$, $\dots$, and $c_s = \int_{x_s}^{b} p(t) \, dt$. Assume $\sum_{i=0}^s c_i \le 4s/n$. Then, for some universal constant $C \ge 1$ there is an algorithm that runs in $O(s)$ time and outputs an interval $I' \subseteq J$ where
\begin{equation*}
    \frac{1}{2} \left(\sqrt{\int_{I'} p(t) \, dt} - \sqrt{\text{(\# of $x_i$ in $I'$)}/n} \right)^2 \ge \sup_{I \subseteq J}\frac{1}{2C} \left(\sqrt{\int_I p(t) \, dt} - \sqrt{\text{(\# of $x_i$ in $I$)}/n} \right)^2 - \frac{C}{n}.
\end{equation*}
\end{theorem}
\begin{proof}
In this theorem, we are replacing the component of the algorithm from Section 6.3.2 of \cite{acharya2017sample}. At this point, the algorithm has already computed these $c_i$ values so we may also use them.

Let us justify the $\sum_{i=0}^s c_i \le 4s/n$ assumption in the theorem. First, we may assume $\tau \le s/n$ in the context of the algorithm, as otherwise we would immediately know the all-zero polynomial is in $\mathcal{C}_\tau$ and we would not need to call this subroutine. Second, we may assume $\int_J p \le 4s/n$, because otherwise we could immediately return the entire interval $J$ as a type (ii) constraint that $p$ violates for $\tau \le s/n$.

At this point, we will reduce the task to a discrete and unweighted problem.  We observe that for any interval $I \subseteq J$, the integral $\int_I p$ is additively within $1/n$ of [the number of times the CDF of $p$ exceeds a multiple of $1/n$ in $I$] multiplied by $1/n$. Call each of these points where the CDF of $p$ exceeds a multiple of $1/n$ as the ``special points.'' Instead of considering the $\mathcal{H}_1$ distance between $p$ and the empirical distribution, we will consider the $\mathcal{H}_1$ distance between the empirical distribution and these special points (each weighted $1/n$ like a sample). In particular, we will find an approximate optimizer for:
\begin{equation}\label{eq:easier-opt}
\max_{I \subseteq J} \frac{1}{2} \left(\sqrt{\text{(\# special points in $I$)}/n} - \sqrt{\text{(\# samples in $I$)}/n} \right)^2
\end{equation}

The value of this quantity for any interval $I$ will be close to the corresponding quantity for the $\mathcal{H}_1$ distance between $p$ and the empirical distribution (they will be within a multiplicative constant factor and an $O(1/n)$ additive term). This follows from the approximate triangle inequality $(\sqrt{a}-\sqrt{b})^2 \lesssim (\sqrt{a}-\sqrt{c})^2 + (\sqrt{c}-\sqrt{b})^2$, where we use $|\text{(\# special points in $I$)}/n - \int_I p(t) \, dt | \le 1/n$.

The form of \Cref{eq:easier-opt} is amenable to being viewed as a discrete and unweighted problem. Let us define $d_i = \left\lfloor \frac{\sum_{j=0}^i c_j}{1/n}\right\rfloor - \left\lfloor \frac{\sum_{j=0}^{i-1} c_j}{1/n}\right\rfloor$; meaning, $d_i$ denotes the number of special points between $x_{i}$ and $x_{i+1}$. Now, consider flattening the samples and the special points into an array $A$, where the samples are $+1$ and special points are $-1$. More formally, $A$ is $d_0$ $-1$'s, then one $+1$, then $d_1$ $-1$'s, then one $+1$, $\dots$, and then $d_s$ $-1$'s. The optimization problem of \Cref{eq:easier-opt} is equivalent to:

\begin{equation}\label{eq:array-opt}
\max_{l\le r \text{ where } l,r\in \{0,\dots,|A|-1\}} \frac{1}{2} \left(\sqrt{\text{(\# of -1's in $A_l,\dots,A_r$)} /n} - \sqrt{\text{(\# of +1's in $A_l,\dots,A_r$)} /n} \right)^2
\end{equation}

Note how since $\int_J p \le 4s/n$, we know there are at most $4s$ special points and thus $|A|\le 5s$. We could thus naively compute \Cref{eq:array-opt} exactly in $O(s^2)$ time. However, we will instead design a procedure running in $O(s)$ time that optimizes \Cref{eq:array-opt} within a constant factor. 

\textit{Algorithm intuition. } In the beginning of the proof of Claim 2.22 in \cite{compton2025attainability}, they show for non-negative numbers $x,y$ that $| \sqrt{x} - \sqrt{y} | \le \frac{2 |x-y|}{\sqrt{\max(x,y)}}$. Further, it is clear $| \sqrt{x} - \sqrt{y} |  = \int_{\min(x,y)}^{\max(x,y)} \frac{1}{2 \sqrt{t}} dt\ge \frac{|x-y|}{2\sqrt{\max(x,y)}}$. Together, this implies $| \sqrt{x} - \sqrt{y} | = \Theta(1) \cdot \frac{|x-y|}{\sqrt{\max(x,y)}} = \Theta(1) \cdot \frac{|x-y|}{\sqrt{x+y}}$. 

This guides our algorithm: if we are able to consider sets of intervals where the $x+y$ values are within a constant factor of each other, then it is sufficient to choose the interval with maximum $|x-y|$ value. This is much more convenient to work with. Consider the following optimization task where we hope to find the interval with largest $|x-y|$ within some larger interval $[l,r]$:

\begin{equation}\label{eq:diff-opt}
    \max_{l',r' \in \{0,\dots,|A|-1\} \text { s.t. } l \le l' \le r' \le r}  \left| (\text{\# of -1's in }A_{l'},\dots,A_{r'}) - (\text{\# of +1's in }A_{l'},\dots,A_{r'}) \right|
\end{equation}

Let us define a prefix sum array $B$, where $B_0 = 0$ and then $B_i = B_{i-1} + A_{i-1}$ for $i \in \{1,\dots,|A|\}$. Then \Cref{eq:diff-opt} is equivalent to:

\begin{equation}\label{eq:prefix-opt}
    \max_{i \in \{l,\dots,r+1\}} B_i - \min_{i \in \{l,\dots,r+1\}} B_i
\end{equation}

Accordingly, finding the maximum absolute value difference within some interval $[l,r]$ reduces to simply computing a range minimum and range maximum query. Our plan is as follows: (i) construct a set of $O(s)$ intervals $[l,r]$, where for any $[l^*,r^*]$ there is a corresponding $[l,r]$ interval that contains $[l^*,r^*]$ and is only a constant factor larger, and (ii) consider the interval optimizing \Cref{eq:diff-opt} for each $[l,r]$, and take the best interval among these optimizers. We remark part (ii) can be done in $O(s)$ time because $O(n)$ range minimum (and maximum) queries on an array of length $n$ can be solved in $O(n)$ time (e.g. via \cite{gabow1984scaling} and \cite{harel1984fast}); the data structure can also be modified to return the index of the optimizer in addition to the value (so we may reconstruct the interval).

\textit{Algorithm. } Let us first construct our set of considered intervals $[l,r]$. Include all intervals with a single element: $[j,j]$ for $j \in \{0,\dots,|A|-1\}$. Further, for all $i \in \{1,\dots,\lceil \log(|A|) \rceil \}$, consider adding intervals of the form $[k \cdot 2^{i-1}, \min(k \cdot 2^{i-1} + 2^i - 1,|A|-1)]$, for all non-negative integer values of $k$ where $k \cdot 2^{i-1} < |A|$. 

Once we have this set of intervals, our algorithm simply considers the intervals given by the range query optimizers of \Cref{eq:prefix-opt}, and chooses the optimizing interval that maximizes \Cref{eq:array-opt}.

\textit{Analysis. } We observe how the number of these intervals is $O(s)$. Additionally, we claim for any interval $[l^*,r^*]$, there is an interval $[l,r]$ in our set where $(r-l+1) \le 4 (r^*-l^*+1)$ (meaning, the length of the interval $[l,r]$ is at most four times the length of $[l^*,r^*]$). Let $i'$ denote the smallest integer where $2^{i'} \ge (r^*-l^*+1)$. Then, the last interval of length $2^{i' + 1}$ starting before $[l^*,r^*]$ must contain this interval, and it will only be at most four times its length. 

Let $[l^*,r^*]$ denote the interval that optimizes \Cref{eq:array-opt}. Let $[l,r]$ denote the interval in our set that is at most four times the length of $[l^*,r^*]$. Finally, let $[l',r']$ denote the optimizer of \Cref{eq:prefix-opt} within $[l,r]$. We now argue the value of the objective for \Cref{eq:array-opt} with $[l',r']$ is at least a constant fraction of the value for $[l^*,r^*]$, which immediately implies our desired result:

\begin{align*}
    &\frac{1}{2} \left(\sqrt{\text{(\# of -1's in $A_{l'},\dots,A_{r'}$)} /n} - \sqrt{\text{(\# of +1's in $A_{l'},\dots,A_{r'}$)} /n} \right)^2 \\
    &\ge \frac{1}{2} \left( \frac{\left| \text{(\# of -1's in $A_{l'},\dots,A_{r'}$)} /n - \text{(\# of +1's in $A_{l'},\dots,A_{r'}$)} /n\right|}{2 \sqrt{\max(\text{(\# of -1's in $A_{l'},\dots,A_{r'}$)} /n,\text{(\# of +1's in $A_{l'},\dots,A_{r'}$)} /n)}} \right)^2 \\
    &\ge \frac{1}{2} \left( \frac{\left| \text{(\# of -1's in $A_{l'},\dots,A_{r'}$)} /n - \text{(\# of +1's in $A_{l'},\dots,A_{r'}$)} /n\right|}{2 \sqrt{\text{(\# of -1's in $A_{l'},\dots,A_{r'}$)} /n + \text{(\# of +1's in $A_{l'},\dots,A_{r'}$)} /n}} \right)^2 \\
    &\ge \frac{1}{2} \left( \frac{\left| \text{(\# of -1's in $A_{l^*},\dots,A_{r^*}$)} /n - \text{(\# of +1's in $A_{l^*},\dots,A_{r^*}$)} /n\right|}{2 \sqrt{\text{(\# of -1's in $A_{l'},\dots,A_{r'}$)} /n + \text{(\# of +1's in $A_{l'},\dots,A_{r'}$)} /n}} \right)^2 \rtag{since $[l',r']$ is optimizer}\\
    &\ge \frac{1}{2} \left( \frac{\left| \text{(\# of -1's in $A_{l^*},\dots,A_{r^*}$)} /n - \text{(\# of +1's in $A_{l^*},\dots,A_{r^*}$)} /n\right|}{4 \sqrt{\text{(\# of -1's in $A_{l^*},\dots,A_{r^*}$)} /n + \text{(\# of +1's in $A_{l^*},\dots,A_{r^*}$)} /n}} \right)^2 \rtag{factor of four length bound}\\
    &\ge \frac{1}{64} \left( \frac{\left| \text{(\# of -1's in $A_{l^*},\dots,A_{r^*}$)} /n - \text{(\# of +1's in $A_{l^*},\dots,A_{r^*}$)} /n\right|}{\sqrt{\max(\text{(\# of -1's in $A_{l^*},\dots,A_{r^*}$)} /n , \text{(\# of +1's in $A_{l^*},\dots,A_{r^*}$)} /n)}} \right)^2 \\
    &\ge \frac{1}{256}  \left| \sqrt{\text{(\# of -1's in $A_{l^*},\dots,A_{r^*}$)} /n} - \sqrt{\text{(\# of +1's in $A_{l^*},\dots,A_{r^*}$)} /n} \right|^2 \rtag{$|\sqrt{x} - \sqrt{y}| \le \frac{2 |x-y|}{\sqrt{\max(x,y)}}$} \\
    & = \frac{1}{128} \cdot \left( \frac{1}{2}  \left| \sqrt{\text{(\# of -1's in $A_{l^*},\dots,A_{r^*}$)} /n} - \sqrt{\text{(\# of +1's in $A_{l^*},\dots,A_{r^*}$)} /n} \right|^2 \right) %
\end{align*}
\end{proof}

\subsection{Combining ingredients}\label{subsec:combine-ingreds}

We plan to invoke the greedy merging algorithm of \Cref{thm:merge-guarantee} with our oracle given by \Cref{thm:p-opt}. This gives us an algorithm that approximately minimizes the maximum $\mathcal{H}_1$ distance between $\hat{f}_n$ and $f_n$ \textit{within any piece}. We will manipulate this estimate $\hat{f}_n$ to give a bound on the distance $\mathcal{H}_1(\hat{f}_n,f_n)$ without restricting the intervals:

\begin{corollary}\label{cor:final-h1-opt}
Consider an interval $J=[a,b]$ and $n$ samples $f_n$ where $a \le x_1 \le \dots \le x_n \le b$. Consider $\mathcal{F}_d^t$: the class of $t$-piece degree-$d$ densities over $J$ (in the case of $d=1$, we may further restrict to lines where the value at the endpoints of $J$ are within a factor of two of each other). There exists an algorithm that runs in time 
\begin{equation*}
        O \left( (d^3 \log \log n + nd^2 + d^{\omega + 2}) \log^3 (n) \right)
\end{equation*}
and outputs an estimate $\hat{f}_n \in \mathcal{F}_d^{4t}$ where
\begin{equation*}
    \mathcal{H}_1(\hat{f}_n,f_n) \lesssim \inf_{f_\theta \in \mathcal{F}_{d}^{2t}} \mathcal{H}_1(f_\theta,f_n) + \frac{1}{n}.
\end{equation*}
\end{corollary}
\begin{proof}
Combining \Cref{thm:merge-guarantee,thm:p-opt} gives us an algorithm with the desired runtime where the maximum $\mathcal{H}_1$ distance between $\hat{f}_n$ and $f_n$ \textit{within any piece} is at most $\lesssim \opt_{2t} + \frac{1}{n}$. However, we aim for an $\mathcal{H}_1$ bound that is not just within pieces, but is over all intervals.

Modify $\hat{f}_n$ by rescaling each piece such that its mass is the same as the empirical mass of $f_n$ in this piece. We will quickly show that this does not increase the maximum $\mathcal{H}_1$ distance between the new $\hat{f}_n$ and $f_n$ by more than a constant factor within any piece. Let $J$ denote the interval corresponding to some piece, and let $I$ denote some interval $I \subseteq J$. We define $w_I(p)$ as the total mass in $I$ for some density $p$. Then, for any $I$ we know the $\mathcal{H}_1$ contribution after rescaling is still upper bounded:
\begin{align*}
    &\frac{1}{2} \cdot \left( \sqrt{\frac{w_J(f_n)}{w_J(\hat{f}_n)} \cdot w_I(\hat{f}_n)} - \sqrt{w_I(f_n)}\right)^2  \\
    & \lesssim \left( \sqrt{w_I(\hat{f}_n)} - \sqrt{w_I(f_n)}\right)^2  + \left( \sqrt{\frac{w_J(f_n)}{w_J(\hat{f}_n)} \cdot w_I(\hat{f}_n)} - \sqrt{w_I(\hat{f}_n)}\right)^2  \\
    &  \lesssim \left( \sqrt{w_I(\hat{f}_n)} - \sqrt{w_I(f_n)}\right)^2  + \left( \sqrt{\frac{w_J(f_n)}{w_J(\hat{f}_n)} \cdot w_J(\hat{f}_n)} - \sqrt{w_J(\hat{f}_n)}\right)^2 \\
    &  = \left( \sqrt{w_I(\hat{f}_n)} - \sqrt{w_I(f_n)}\right)^2  + \left( \sqrt{w_J(f_n)} - \sqrt{w_J(\hat{f}_n)}\right)^2 \\
    & \lesssim \opt_{2t} + \frac{1}{n}
\end{align*}
The last step followed from how the $\mathcal{H}_1$ distance between $\hat{f}_n$ and $f_n$ within each piece, before rescaling, was $\lesssim \opt_{2t} + \frac{1}{n}$. Hence, after rescaling, it still holds that the $\mathcal{H}_1$ distance within pieces is $\lesssim \opt_{2t} + \frac{1}{n}$.

Now, observe how any interval $I$ that overlaps with multiple pieces can be described as: a suffix of some piece, followed by some number of entirely contained pieces, followed by a prefix of some piece. Recall how we rescaled $\hat{f}_n$ so the integral over each piece is the same as the number of samples. Thus, if we use the upper bound on the $\mathcal{H}_1$ objective for $I$ given by \Cref{fact:dpi}, then the completely-contained pieces contribute $0$. This means the $\mathcal{H}_1$ objective for $I$ is at most a factor of two larger than the maximum $\mathcal{H}_1$ for intervals within a single piece. In total, this proves that after rescaling $\mathcal{H}_1(\hat{f}_n,f_n) \lesssim \inf_{f_\theta \in \mathcal{F}_{d}^{2t}} \mathcal{H}_1(f_\theta,f_n) + \frac{1}{n}$.
\end{proof}

We are now ready to prove \Cref{thm:fast-log-concave}:
\begin{proof}
For log-concave densities, we will use \Cref{cor:final-h1-opt} with $\lin_{2t}$. Since $d=1$, this gives us an algorithm running in $O(n \log^3 n)$ time with the stated upper bound on $\mathcal{H}_1(\hat{f}_n,f_n)$. Thus, we may use \Cref{lemma:h1-to-err} to conclude the estimate $\hat{f}_n$ has error at most
\begin{equation*}
    \dhsq(f,\hat{f}_n) \lesssim \inf_{f_\theta \in \lin_t} \dhsq(f,f_\theta) + t \cdot \frac{\log(n) + \log(2/\delta)}{n}.
\end{equation*}

Since we may approximate each mixture component with $t/k$ pieces, we use \Cref{lemma:piece-linear} to imply:
\begin{equation*}
    \inf_{f_\theta \in \lin_t} \dhsq(f,f_\theta) \lesssim \inf_{f_\theta \in \lin_t, g_\theta \in \mathcal{LC}_k} \dhsq(f,g_\theta) + \dhsq(f_\theta,g_\theta) \lesssim \inf_{g_\theta \in \mathcal{LC}_k} \dhsq(f,g_\theta) + \left( \frac{k}{t} \right)^4
\end{equation*}

Together these imply:

\begin{align*}
    \dhsq(f,\hat{f}_n) &\lesssim  \inf_{f_\theta \in \mathcal{LC}_k} \dhsq(f,f_\theta)  + \left( \frac{k}{t} \right)^4 + t \cdot \frac{\log(n) + \log(2/\delta)}{n} \\
    & \lesssim  \inf_{f_\theta \in \mathcal{LC}_k} \dhsq(f,f_\theta)  + \left( \frac{k \cdot (\log(n) + \log(2/\delta))}{n}\right)^{4/5} \rtag{choose $t = \left(\frac{nk^4}{\log(n) + \log(2/\delta)}\right)^{1/5}$}
\end{align*}
\end{proof}

\section{Details for piecewise polynomial results}\label{app:piecewise-poly}

This section will mostly follow from the tools of \Cref{app:log-concave}. We first state an analogous result to \Cref{lemma:h1-to-err}: we show that an $\mathcal{H}_1$ bound between $\hat{f}_n$ and $f_n$ implies a squared Hellinger bound between $\dhsq(f,\hat{f}_n)$. The only technical difference is in this setting with degree-$d$ polynomials, we no longer have a stronger custom reverse data processing inequality, so extra logarithmic terms appear.

\begin{lemma}\label{lemma:poly-h1-to-err}
     Suppose we receive $n$ samples from the distribution $f$, and that $\hat{f}_n \in \mathcal{F}_d^{4t}$ is a density where $\mathcal{H}_1(\hat{f}_n,f_n) \le \alpha \cdot \inf_{f_\theta \in \mathcal{F}_d^{2t}} \mathcal{H}_1(f_\theta,f_n) + \beta$ for $\alpha \ge 1$ and $\beta > 0$. Then with probability $1-\delta$,
    \begin{equation*}
        \dhsq(f,\hat{f}_n) \lesssim  \inf_{f_\theta \in \mathcal{F}_d^t}\alpha d\dhsq(f,f_\theta) \log(\alpha t e / \dhsq(f,f_\theta)) + \alpha dt \cdot \frac{\log(n) + \log(2/\delta)}{n} \cdot \log(\alpha e n) +  \beta dt \log(\alpha e / \beta)
    \end{equation*}
\end{lemma}

\begin{proof}
    Consider any $f_\theta \in \mathcal{F}_d^t$. It holds that
    \begin{align*}
        & \dhsq(f,\hat{f}_n) \lesssim  \dhsq(f,f_\theta) + \dhsq(f_\theta,\hat{f}_n) \intertext{Observe how for any two densities in $\mathcal{F}_d^{t}$ and $\mathcal{F}_d^{4t}$, the domain can be decomposed into $\lesssim t$ intervals where within each interval both densities are one-piece functions in $\mathcal{F}_d$. Further, for each of these pieces and for any value $c>0$, the regions where $f_\theta(x) \ge c \cdot \hat{f}_n(x)$ consist of at most $\lesssim d$ intervals (this is because it can be written as the region where a degree-$d$ polynomial is non-negative, which only changes at its $\le d$ roots). Accordingly, this characterizes the ratio class $\mathcal{H}$ as a subset of the union of at most $rdt$ intervals (for some constant $r$),  and by \Cref{thm:rdpi} and \Cref{eq:move-log} we observe:}
        & \lesssim  \dhsq(f,f_\theta) + \hdist(f_\theta,\hat{f}_n)  \log(e/\hdist(f_\theta,\hat{f}_n)) \intertext{As in the proof of \Cref{thm:hel-min-dist}, recall that $z(x) \triangleq x \log(e/x)$ is non-decreasing for $x\in[0,1]$, and $z(x) \le 2(z(a)+z(b))$ if both $0 \le x,a,b \le 1$ and $x \le 2(a+b)$.}
        & \lesssim  \dhsq(f,f_\theta) + \mathcal{H}_{rdt}(f_\theta,\hat{f}_n) \log(e/\mathcal{H}_{rdt}(f_\theta,\hat{f}_n)) \\
        & \lesssim  \dhsq(f,f_\theta) \log(e/\dhsq(f,f_\theta)) + \mathcal{H}_{rdt}(f,\hat{f}_n) \log(e/\mathcal{H}_{rdt}(f,\hat{f}_n)) \rtag{via $\mathcal{H}_{rdt}(a,b) \le \dhsq(a,b)$}\\
        & \lesssim   \dhsq(f,f_\theta) \log(e/\dhsq(f,f_\theta)) + dt\mathcal{H}_{1}(f,\hat{f}_n) \log(e/\mathcal{H}_{1}(f,\hat{f}_n))\rtag{using $\mathcal{H}_{rdt}(a,b) \le rdt\mathcal{H}_1(a,b)$}\\
        & \lesssim   \dhsq(f,f_\theta) \log(e/\dhsq(f,f_\theta)) + dt\mathcal{H}_{1}(f,f_n) \log(e/\mathcal{H}_{1}(f,f_n))   + dt\mathcal{H}_{1}(f_n,\hat{f}_n) \log(e/\mathcal{H}_{1}(f_n,\hat{f}_n)) \label{eq:poly-intermediate} \numberthis
    \end{align*}
    At this point, we would like to design a modification of $f_\theta$ where each piece of the modified version has the same integral as $f$ over its domain, and the Hellinger distance contribution of each piece is roughly evenly-split. Later, this modification will be helpful because it will have small $\mathcal{H}_1$ distance from $f$. Applying \Cref{lemma:tilde-f} to $f_\theta$, we obtain $\tilde{f}_\theta$ where $\tilde{f}_\theta \in \mathcal{F}_d^{2t-1}$, $\dhsq(f,\tilde{f}_\theta) \lesssim \dhsq(f,f_\theta)$, and $\max_{\mathcal{I} \text{ within a piece of }\tilde{f}_\theta} \dhsq(f_\mathcal{I},(\tilde{f}_\theta)_\mathcal{I})\lesssim \frac{\dhsq(f,\tilde{f}_\theta)}{t}$.

    Let $z_\alpha(x) \triangleq x \log(\alpha e /x)$. Similarly, we conveniently use  (i) $z_\alpha(x)$ is non-decreasing for $x \in [0,\alpha]$, (ii)  $z_\alpha(x) \le 2(z_\alpha(a)+z_\alpha(b))$ if both $0 \le x , a, b \le \alpha$ and $x \le 2(a+b)$, and (iii) $z_\alpha(x) \le z_\alpha(a) + z_\alpha(b)$ if both $0 \le x , a, b \le \alpha$ and $x \le a+b$. Continuing from \Cref{eq:poly-intermediate}, we use how the distance $\mathcal{H}_1(f_n,\hat{f}_n)$ is bounded in terms of any other density in $\mathcal{F}_d^{2t}$, including $\tilde{f_\theta}$:
    \begin{align*}
        & = z(\dhsq(f,f_\theta)) + dt \cdot z(\mathcal{H}_{1}(f,f_n)) + dt \cdot z(\mathcal{H}_{1}(f_n,\hat{f}_n))\\
        &\le z_\alpha(\dhsq(f,f_\theta)) + dt \cdot z_\alpha(\mathcal{H}_{1}(f,f_n)) + dt \cdot z_\alpha(\mathcal{H}_{1}(f_n,\hat{f}_n))\\
        & \lesssim  z_\alpha(\dhsq(f,f_\theta)) + dt \cdot z_\alpha(\mathcal{H}_{1}(f,f_n)) + \alpha dt \cdot z_\alpha(\mathcal{H}_{1}(f_n,\tilde{f}_\theta)) + dt \cdot z_\alpha(\beta)\\
        & \lesssim  z_\alpha(\dhsq(f,f_\theta)) + \alpha dt \cdot z_\alpha(\mathcal{H}_{1}(f,f_n)) + \alpha dt \cdot z_\alpha(\mathcal{H}_{1}(f,\tilde{f}_\theta)) + dt \cdot z_\alpha(\beta) \intertext{Leveraging \Cref{fact:dpi} and how the integral of $f$  and $\tilde{f_\theta}$ is the same for any piece of $\tilde{f_\theta}$, we observe the distance $\mathcal{H}_1(f,\tilde{f_\theta})$ is at most a constant factor larger than the same $\mathcal{H}_1$ distance when we restrict to intervals within one piece of $\tilde{f_\theta}$. This quantity is at most the maximum squared Hellinger distance between $f$ and $\tilde{f_\theta}$ over one piece. By the properties of $\tilde{f_\theta}$ given by \Cref{lemma:tilde-f}, this implies a bound of $\mathcal{H}_1(f,\tilde{f_\theta}) \lesssim \frac{1}{t} \cdot \dhsq(f,f_\theta)$:}
        & \lesssim  z_\alpha(\dhsq(f,f_\theta)) + \alpha dt \cdot z_\alpha(\mathcal{H}_{1}(f,f_n)) + \alpha dt \cdot z_\alpha(\dhsq(f,f_\theta)/t) + dt \cdot z_\alpha(\beta) \\
        & \lesssim \alpha d\dhsq(f,f_\theta) \log(\alpha t e / \dhsq(f,f_\theta)) + \alpha dt \mathcal{H}_{1}(f,f_n) \log(\alpha e / \mathcal{H}_{1}(f,f_n)) +  \beta dt  \log(\alpha e / \beta) \\
        & \lesssim \alpha d\dhsq(f,f_\theta) \log(\alpha t e / \dhsq(f,f_\theta)) + \alpha dt \cdot \frac{\log(n) + \log(2/\delta)}{n} \cdot \log(\alpha e n) +  \beta dt \log(\alpha e / \beta)
    \end{align*}
    The last line holds with probability $1-\delta$ from the uniform convergence guarantee \Cref{cor:sqrt-converge}.
\end{proof}

Together, \Cref{lemma:poly-h1-to-err} and \Cref{cor:final-h1-opt} immediately imply \Cref{cor:poly-final-h1-opt}.

\section{Proof of \Cref{corr:fast-gauss}}\label{app:fast-poly-corr}
\begin{proof}
    Given \Cref{cor:poly-final-h1-opt}, the proof will follow from the same logic as \Cref{cor:stat-gaussian}. 

    We will choose $\mathcal{F}$ as the class of $3k$-piecewise degree $\lesssim \log(n)$ polynomial densities. By \Cref{cor:poly-final-h1-opt}:
    \begin{align*}
        &\dhsq(f,\hat{f}_n) \lesssim  \inf_{f_\theta \in \mathcal{F}_d^t} \log(n) \cdot \dhsq(f,f_\theta) \log(2k / \dhsq(f,f_\theta)) + \frac{k \log^2(n) \cdot (\log(n) + \log(2/\delta))}{n} \intertext{We use existing results that imply any Gaussian density is approximated within $\eps$ squared Hellinger distance by a 3-piecewise degree $\lesssim \log(1/\eps)$ polynomial (Lemma 36 of \cite{chan2014efficient}, or related discussion in Section 5.4 of \cite{timan}). Accordingly, for any mixture of $k$ Gaussians, there is a $3k$-piecewise degree $\lesssim \log(n)$ polynomial that approximates the mixture within $1/n$ squared Hellinger distance: }
        & \lesssim  \inf_{f_\theta \in \mathcal{G}_k} \log(n) \cdot \dhsq(f,f_\theta) \log(2k / \dhsq(f,f_\theta)) + \frac{1}{n} + \frac{k \log^2(n) \cdot (\log(n) + \log(2/\delta))}{n} \\
        & \lesssim  \inf_{f_\theta \in \mathcal{G}_k} \log(n) \cdot  \dhsq(f,f_\theta) \log(2k / \dhsq(f,f_\theta)) + \frac{k \log^2(n) \cdot (\log(n) + \log(2/\delta))}{n}
    \end{align*}
    
\end{proof}
\end{document}